\documentclass[format=acmsmall, natbib, review=false]{acmart}

\makeatletter
\def\@ACM@checkaffil{% Only warnings
    \if@ACM@instpresent\else
    \ClassWarningNoLine{\@classname}{No institution present for an affiliation}%
    \fi
    \if@ACM@citypresent\else
    \ClassWarningNoLine{\@classname}{No city present for an affiliation}%
    \fi
    \if@ACM@countrypresent\else
        \ClassWarningNoLine{\@classname}{No country present for an affiliation}%
    \fi
}
\makeatother

\usepackage{acm-style}
\setcitestyle{numeric}

\usepackage[utf8]{inputenc}
\usepackage{wrapfig}
\usepackage{latexsym}
\usepackage{graphicx,amsmath,amsfonts}
\usepackage{color}
\usepackage{tikz}
\usepackage{paralist}
\usepackage{color,xspace}
\usepackage{comment}
\usepackage{booktabs}
\usepackage{tabularx}
\usepackage{floatflt}
\usepackage[lined, ruled, noend]{algorithm2e}
\usepackage{graphics}
\usepackage{pifont}

\newtheorem{theorem}{Theorem}

\newtheorem{corollary}[theorem]{Corollary}
\newtheorem{lemma}[theorem]{Lemma}

\newtheorem{remark}{Remark}
\newtheorem{definition}{Definition}
\newtheorem{example}{Example}

\newcommand{\ie}{i.e.,\xspace}
\newcommand{\eg}{e.g.,\xspace}

\newcommand{\secref}[1]{Section~\ref{#1}}
\newcommand{\figref}[1]{Figure~\ref{#1}}

\newcommand{\lemref}[1]{Lemma~\ref{#1}}
\newcommand{\thmref}[1]{Theorem~\ref{#1}}

\newcommand{\defref}[1]{Definition~\ref{#1}}

\newcommand{\exref}[1]{Example~\ref{#1}}

\newcommand{\milpmax}{\textsf{P}}

\newcommand{\set}[1]{\{#1\}}
\newcommand{\midd}{\mathrel{:}}

\newcommand{\maxP}{leximax-\phrag} %
\newcommand{\varP}{var-\phrag}
\newcommand{\seqP}{seq-\phrag}
\newcommand{\eneP}{Eneström-\phrag}

\newcommand\bxell{\bar x_{(\ell)}}

\title{Phragm\'{e}n's Voting Methods and Justified Representation}

\author{Markus Brill}
 \orcid{0000-0001-9509-7017}
 \affiliation{%
    \institution{TU Berlin} 
    \country{Germany}}
 \affiliation{%
    \institution{University of Warwick} 
    \country{UK}
    }
 \email{markus.brill@warwick.ac.uk}
\author{Svante~Janson}
 \orcid{0000-0002-9680-2790}
\affiliation{%
    \institution{Uppsala University} 
    \country{Sweden}
    }
 \email{svante.janson@math.uu.se}
\author{Rupert~Freeman}
 \orcid{0000-0003-4744-9449}
 \affiliation{%
    \institution{University of Virginia} 
    \country{USA}
    }
 \email{freemanr@darden.virginia.edu}
\author{Martin~Lackner}
 \orcid{0000-0003-2170-0770}
 \affiliation{%
    \institution{TU Wien} 
    \country{Austria}
    }
 \email{lackner@dbai.tuwien.ac.at}

\newcommand{\phrag}{Phragm\'{e}n\xspace}
\newcommand{\opti}{optimization\xspace}
\newcommand{\Opti}{Optimization\xspace}

\begin{abstract}
In the late 19th century, Swedish mathematician Edvard \phrag proposed a load-balancing approach for selecting committees based on approval ballots. We consider three committee voting rules resulting from this approach: two \opti variants---one {minimizing} the maximum load and one minimizing the variance of loads---and a sequential variant.
We study \phrag's methods from an axiomatic point of view, focusing on 
properties capturing proportional representation.
We show that the sequential variant satisfies \emph{proportional justified representation}, 
which is a rare property for committee monotonic methods.
Moreover, we show that the \opti variants satisfy \emph{perfect representation}. We also analyze the computational complexity of \phrag's methods and provide mixed-integer programming based algorithms for computing them.
\end{abstract}

\begin{document}

\begin{titlepage}

\maketitle

\vspace{2cm}

\begin{center}
   \includegraphics[width=0.5\textwidth]{p.png} 
   
   Lars Edvard \phrag (1863--1937)
\end{center} 

\end{titlepage}

\section{Introduction}

While most of the social choice literature is focused on single-winner scenarios,
recent years have witnessed an increasing interest in \emph{committee voting rules} \citep[\eg][]{elk-fal-sko-sli:c:multiwinner-rules,FSST17a,owaWinner,LaSk21b}. In this setting, a fixed-size subset of alternatives has to be selected based on the preferences of a group of voters.   
In this paper, we assume that the preferences of individual voters are given by \emph{approval ballots}, specifying which alternatives are ``approved'' by the voters. 
For an overview of research on approval-based committee elections, we refer to the recent survey by \citet{LaSk21b}.

A crucial issue in group decision making is \emph{(proportional) representation}. Informally speaking, an outcome of a decision-making process is representative if it reflects the preferences of the members of the group. 
In the context of approval-based committee elections, reasoning about representation is non-trivial. Since approval sets may overlap arbitrarily, there are many different ways in which the set of voters can be split into more or less ``cohesive'' subgroups. 
Whether a given subgroup has a justified claim to be represented in the committee depends on the size of the subgroup as well as on its level of cohesiveness.

\citet{justifiedRepresentation} and \citet{SFF+17a} have identified axiomatic properties capturing the intuitive notion that subgroups that are ``large enough'' and ``cohesive enough'' deserve to be represented in the committee:
\emph{justified representation (JR)}, \emph{proportional justified representation (PJR)}, and \emph{extended justified representation (EJR)}. While a number of standard committee voting rules have been shown to satisfy the basic requirement of JR, it turns out that the more demanding properties PJR and EJR are much harder to satisfy.

In this paper, we consider committee voting rules that are due to Swedish mathematician Edvard \phrag{} (we provide brief biographical information in \secref{sec:bio}).
\phrag phrases committee elections as load balancing problems: 
Adding a candidate to the committee incurs some \emph{load}, and this load should be shared among the voters approving this candidate. \phrag suggests choosing committees in such a way that the corresponding load distributions are as \emph{balanced} as possible, and different ways of measuring balancedness
result in different optimization objectives. This approach yields two \emph{\opti} variants, one minimizing the maximum load and one minimizing the variance of loads, and one \emph{sequential} variant, which proceeds by greedily selecting candidates so as to keep the maximum load as small as possible. In addition to the load balancing rules, \phrag also proposed a rule that adapts the principle behind Single Transferable Vote (STV) to approval ballots.   

Although \phrag's methods were proposed in the same era as \textit{Proportional Approval Voting (PAV)},\footnote{Proportional Approval Voting is a prominent committee voting rule due to Danish polymath Thorvald N.\ Thiele~\citep{Thie95a}. For a detailed comparison between PAV and \phrag's methods, we refer to 
% the works of 
\citet{Jans16a} and \citet{peters2019proportionality}.} they have received hardly any attention until very recently.
Since the publication of the conference version of this paper~\cite{BrillFJL17-phragmen} in 2017,
\phrag{}'s methods became increasingly central in the analysis of approval-based committee rules.\footnote{Two notable studies that predate the conference version of this paper are a  survey by \citet{Jans12a} (in Swedish) and a paper by \citet{MoOl15a} (in Catalan).}
In politics, variants of both \phrag's methods and PAV have been used in Swedish parliamentary elections (for distribution of seats within parties), and a version of one of \phrag's methods is still part of the election law, although in a minor role~\citep{Jans16a}.
Further, \phrag's sequential method is often used for the selection of ``validators'' who participate in a blockchain consensus protocol: 
In the recently introduced \textit{nominated proof-of-stake (NPoS)} mechanism, members of a blockchain community can nominate other members to become validators, and the selection of a representative set of validators plays an important role for the security of the blockchain~\citep{Polk21a,polkadot-overview,CeSt21a}.

\subsection{Results and Outline of the Paper}
After briefly reviewing related work in \secref{sec:related} and introducing some basic notation in \secref{sec:prelim}, we formally define \phrag's methods in \secref{sec:phrag-methods}. 
In \secref{sec:computation}, we analyze the computational complexity of  \phrag's methods and we provide algorithms for computing them. The algorithms for the \opti variants are based on mixed-integer linear and quadratic programming.
In \secref{sec:representation}, we consider the representation axioms mentioned above. We show that the sequential variant satisfies PJR, making it one of few committee monotonic methods with this property. Moreover, we show that the \opti variants satisfy \emph{perfect representation (PR)},
a further representation axiom introduced by \citet{SFF+17a}. The latter result provides a contrast to PAV, which is known to violate~PR.  
In \secref{sec:apportionment}, we discuss the relation between \phrag's methods and  the apportionment problem~\citep{BaYo82a}.

\subsection{A Brief Biography of \phrag}\label{sec:bio}

Lars Edvard \phrag (1863--1937)
was a Swedish mathematician, actuary and insurance executive.
He began his mathematical university studies in Uppsala in 1882, but
transferred in 1883 to Stockholm, where he became a student (and later confidant) of G\"osta Mittag-Leffler \cite{stubhaug2010gosta}.
In 1888, \phrag was appointed coeditor 
of Mittag-Leffler's journal \emph{Acta Mathematica}, where he immediately made
an important contribution by finding an error in a paper by Henri Poincar\'e on the
three-body problem. The paper had been awarded a prize in a competition that
Mittag-Leffler had persuaded King Oscar II to arrange, but \phrag found a
serious mistake when the journal had already been printed; the copies that had
been released were recalled and a new corrected version was printed.

In 1892, \phrag became a professor of mathematics at Stockholm University.
In 1897, he additionally became an actuary in 
a private insurance company. 
His  interest in actuarial science and insurance companies appears to have grown in these years,
as in 1904 he left his professorship to become the first 
head of the Swedish Insurance Supervisory Authority.
In 1908 he became director of a private insurance company, which he remained until 1933.
His involvement in mathematics is witnessed, e.g., by his attendance at the 1924 International Mathematical Congress in Toronto, where he was elected one of the vice-presidents of the International Mathematical Union \cite{imc-toronto}. \phrag also continued to be an editor of Acta Mathematica until his death in
1937.

His best known mathematical work is the \phrag-Lindel\"of principle in complex
analysis, a joint work with Finnish mathematician Ernst Lindel\"of~\cite{phragmen1908extension}.
His interest in election methods is witnessed by his publications
\cite{Phrag93,Phra94a, Phra95a, Phra96a, Phra99a}.
Moreover,
he was a member of the \textit{Royal Commission on a Proportional Election Method
1902--1903} and of a new \textit{Royal Commission on the Proportional Election Method
1912--1913}. For further information we refer the reader to the survey by \citet{Jans16a}
and to the book by \citet{stubhaug2010gosta} (in particular for his relation with Mittag-Leffler).

\section{Related Work}
\label{sec:related}

Proportional representation is an important issue in committee voting (see the influential paper by \citet{Monr95a} and the references therein) and methods ensuring representation often lead to interesting computational problems~\citep{PoBr98a,PSZ08a,LuBo11d,BSU13a}.

The problem of choosing representative committees based on approval ballots can be seen as a generalization of the classical \textit{apportionment} problem \cite{BaYo82a}. The latter setting corresponds to the special case in which candidates are arranged into party lists and each voter chooses a single list; see \secref{sec:apportionment} for details. 
Voting settings between apportionment and approval-based committee voting have also been studied~\citep{BGP+22a}.

For the setting of approval-based committee voting~\citep{Kilg10a,LaSk21b}, \citet{justifiedRepresentation} proposed two representation axioms: \emph{justified representation (JR)} and its strengthening \emph{extended justified representation (EJR)}. 
Later, \citet{SFF+17a} observed that EJR is not compatible with what they call \emph{perfect representation} and proposed an axiomatic property, \emph{proportional justified representation (PJR)}, that is compatible. EJR implies PJR, which in turn implies JR.  

\citet{justifiedRepresentation} and \citet{SFF+17a} showed that most common committee voting rules fail EJR and PJR. A notable exception is Thiele's PAV~\citep{Thie95a}, which satisfies EJR (and thus PJR). Interestingly, variants of PAV based on different weight vectors fail both EJR and PJR (and even weaker proportionality requirements) \cite{justifiedRepresentation,BLS18a}.
Moreover, a greedy approximation algorithm for PAV known as \emph{sequential PAV} or \emph{reweighted approval voting} fails JR (and consequently PJR and EJR) \cite{justifiedRepresentation,SFF+17a}.

Computing the outcome of PAV is NP-hard~\citep{owaWinner,AGG+15a} and thus not feasible in polynomial time unless $\text{P}=\text{NP}$. Prior to our work, it had remained an open question whether there exist polynomial-time computable rules satisfying EJR or PJR. \phrag{}'s sequential rule, as we show in this paper, is polynomial-time computable and satisfies~PJR.

Recent work has established that even EJR can be guaranteed by a polynomial-time voting rule.
This was first shown by \citet{AEHLSS-polyejr}.
Later, \citet{peters2019proportionality} presented the Method of Equal Shares (MES), which is also polynomial-time computable and satisfies EJR. Interestingly, MES is based on the same principle as \phrag's sequential method and shares some of its desirable properties (such as laminar proportionality and priceability \cite{peters2019proportionality}).
None of these rules, however, are committee monotonic,\footnote{A committee voting rule is \emph{committee monotonic} if increasing the committee size results in a winning committee that is a superset of the previously winning committee.
An example showing that MES violates committee monotonicity can be found in the survey by~\citet{LaSk21b}. We are not aware of a formal proof that the rules by \citet{AEHLSS-polyejr} fail committee monotonicity, but the way they are defined makes this claim very plausible.} i.e., an increase in the committee size by one may result in a completely different committee.
In many settings, committee monotonicity is highly desirable (e.g., when generating rankings~\citep{SLB+17a,IsBr21b,RST22a}), and thus \phrag's sequential method---which is committee monotonic by definition---has gained much attention in recent years.
\phrag's sequential method also satisfies further monotonicity axioms~\cite{Jans16a,SaFi17a}.

The maximin support method, introduced by \citet{SFFB22b}, is closely related to \phrag's sequential method and shares many of its axiomatic properties (including PJR and committee monotonicity). The optimization variant of the maximin support method coincides with one of the optimization variants of \phrag's methods, and yields an equivalent formulation of the latter in terms of maximin support~\citep{SFFB22b}. 
An interesting distinction between \phrag's sequential rule and the maximin support method concerns their ability to approximate the optimal solution of the maximin support problem~\citep{CeSt21a}.

Proportional representation has also been studied in settings where voters have ordinal preferences over candidates~\citep{elk-fal-sko-sli:c:multiwinner-rules,FSST17a} and in \textit{participatory budgeting}, a generalization of committee elections where candidates have costs and the set of selected candidates needs to satisfy a budget constraint~\citep{AZSh20a,PPP21a}.
Different variants of \phrag's methods have been generalized to those settings~\citep{Jans16a,aziz2018proportionally,AzLe20a}. Further generalizations of \phrag's methods have been considered in the context of degressive and regressive proportionality~\citep{JaSk22a} and in the context of perpetual voting~\citep{LaMa23a}.

\section{Preliminaries}
\label{sec:prelim}

We consider a social choice setting with a finite set $N=\{1,\ldots, n\}$ of
\emph{voters} and a finite set $C$ of \emph{candidates}. Throughout the paper we let $m =
|C|$ denote the number of candidates and $n=|N|$ the number of voters. The
preferences of each voter $i\in N$ are given by a subset $A_i\subseteq C$, representing the subset of candidates that the voter approves
of. We refer to the list $A = (A_1,\ldots, A_n)$ as the {\em preference
  profile}.
For a candidate $c \in C$, we let $N_c$ denote the set of voters approving $c$,
\ie  $N_c = \{i\in N \midd c \in A_i\}$.
To avoid trivialities, we assume that $N_c\neq\emptyset$ for all $c \in C$.
 
We want to select a subset consisting of exactly $k$ candidates, for a given natural number $k \le m$. 
An {\em approval-based committee voting rule} (henceforth simply
\emph{rule}) maps an instance $(A, k)$ to a subset $S \subseteq C$ of
size~$k$, the {\em committee}. 
In general, there may be ties, and we then allow the rule to yield several
choices, so formally the rule is a map from instances to non-empty sets of
committees. 

Finally, for a tuple of real numbers $z=(z_1,\dots,z_n)$, we let $z_{(\ell)}$ denote the $\ell$-th largest element in~$z$,
so that $z_{(1)} \ge z_{(2)} \ge \dots \ge z_{(n)}$.

\section{\phrag's Methods}
\label{sec:phrag-methods}

The main idea behind \phrag's methods is to identify committees whose ``support'' is distributed as evenly as possible among the electorate. \phrag used different formulations for explaining his methods; we refer the reader to the survey by \citet{Jans16a} for an overview and more details. In this paper, we adopt the formulation from the 1899 paper~\citep{Phra99a}.
In this formulation, every candidate in the committee is thought of as incurring one unit of ``\emph{load},'' and the load incurred by candidate $c$ needs to be distributed among the voters in $N_c$. The goal is to find a committee of size $k$ for which the corresponding load distribution is as balanced as possible.

Formally, a \emph{load distribution} is a two-dimensional array $x = (x_{i,c})_{i \in N, c \in C}$ satisfying the following four constraints:
\begin{align}
& 0 \le x_{i,c} \le 1 &\text{for all $i \in N$ and $c \in C$}\label{eq:optimal-cond3a}\\
& x_{i,c} = 0 &\text{ if $c\notin A_i$}\label{eq:optimal-cond3}\\
&\sum_{i\in N} \sum_{c\in C} x_{i,c} = k &\label{eq:optimal-cond1}\\
&\sum_{i\in N} x_{i,c} \in \{0,1\} &\text{ for all $c\in C$}\label{eq:optimal-cond2}
\end{align}
Here, $x_{i,c}$ corresponds to the load that voter $i$ receives from candidate $c$. Constraint \eqref{eq:optimal-cond3} ensures that the load incurred by candidate $c$ is distributed among voters in $N_c$ only, and constraints~(\ref{eq:optimal-cond1}) and~(\ref{eq:optimal-cond2}) ensure that $x$ corresponds to a size-$k$ committee $\{c \in C \midd \sum_{i \in N} x_{i,c} = 1\}$.
 
For a load distribution $x$, we let $\bar x_i$ denote the total load of voter $i \in N$, \ie $\bar x_i=\sum_{c\in C} x_{i,c}$, and we refer to $(\bar x_1, \ldots, \bar x_n)$ as the vector of \emph{voter loads}. Using this notation, constraint~(\ref{eq:optimal-cond1}) reads $\sum_{i\in N} \bar x_i = k$. Note that constraint~(\ref{eq:optimal-cond1}) implies that the \emph{average} voter load is $\frac{k}{n}$.

There are different ways of measuring how \emph{balanced} a given load distribution is, each giving rise to a different optimization objective. One such objective is to minimize the maximum load assigned to a voter, \ie $\min_x \max_{i\in N}\bar x_i$. (This is equivalent to minimizing the maximum difference between a voter load and the average voter load.) Obviously, the average voter load $\frac{k}{n}$ is a lower bound on the maximum voter load, and we call a load distribution $x$ \emph{perfect} if $\bar x_i = \frac{k}{n}$ for all $i \in N$. Another objective is to minimize the \emph{variance} of voter loads, \ie the sum of squared distances from the average voter load. Again, a perfect load distribution is optimal for this objective. 

We further distinguish between ``\opti'' methods, where we solve a global optimization problem to find a load distribution optimizing the objective, and ``sequential'' methods, where we iteratively construct a load distribution, in each round greedily choosing a candidate optimizing the objective at that iteration. 

In this paper, we focus on three rules: the \opti methods \emph{\maxP} and \emph{\varP}---minimizing the maximum voter load and the variance of voter loads, respectively---and the sequential method \emph{\seqP}, which greedily minimizes the maximum voter load. 
For completeness, we also consider the \mbox{\emph{\eneP}} method (see \secref{sec:enestroem}).

The method \seqP was introduced by \phrag in several papers
\citep{Phra94a,Phra95a,Phra96a,Phra99a}, and it is the
variant that he proposed to be used in actual elections. 
\phrag defined this method as a generalization of D'Hondt's apportionment method to the case without party lists (see \secref{sec:apportionment}). 
\Opti variants and the objective of minimizing the variance are discussed in the 1896 paper~\citep{Phra96a}.

\subsection{\Opti Variants}
\label{sec:direct}

We start by defining the \opti variants.
The first optimization variant selects committees corresponding to load distributions minimizing the maximum voter load. 
In case that two or more committees have the same (minimal) maximum load, we employ a specific way of breaking ties.
This is because it might be the case that for two load distributions $x$ and $y$, although $\max_{i\in N}\bar x_i=\max_{i\in N}\bar y_i$, one load distribution is clearly preferable to the other.
\begin{example}\label{ex:ex1}
Let $C = \{ a, b, c \}$, $k=2$, and 
$A=(\{a\},\{a\},\{b\},\{c\})$. Any committee of size 2 contains either $b$ or $c$, which are approved by only one voter each, so the maximum load is 1 for all committees.
However, the committees containing $a$ represent three voters, while the committee $\{ b, c \}$ only represents two.
\end{example}

In order to refine the set of winning committees, we compare two vectors of voter loads according to the \textit{leximax} ordering.\footnote{The leximax ordering is defined analogously to the more commonly used \textit{leximin} ordering (see, \eg \citet{Moul88a}, Definition 1.1). 
In the literature, the leximax ordering is referred to as ``lexicographic minimax'' by \citet{Ogry97a} and as ``lexicographical'' by \citet{Schm69a}.}

\begin{definition}\label{def:leximax}
For $y=(y_1,\dots,y_n)$ and $z=(z_1,\dots,z_n)$, $y$ is \emph{leximax-smaller} than~$z$, denoted $y \mathbin{\dot<} z$, if there exists $j\leq n$ such that $y_{(j)} < z_{(j)}$ and $y_{(i)} = z_{(i)}$ for all $i\leq j-1$.
\end{definition}

We are now ready to define the first optimization variant. 

\smallskip
\noindent \textbf{\maxP:}
The rule \maxP selects all committees corresponding to load distributions $x$  such that $(\bar x_{1},\dots, \bar x_{n})$ is leximax-optimal, \ie minimal with respect to $\mathbin{\dot<}$.

\smallskip
As we will see in \secref{sec:max-results}, leximax tie-breaking is necessary in order to guarantee strong representation properties.

\smallskip
The second optimization variant is based on a different optimization objective. 

\smallskip
\noindent \textbf{\varP:}
The rule \varP selects all committees corresponding to load distributions minimizing $\sum_{i\in N} \bar x_i^{\,2}$. 

\smallskip 
Minimizing $\sum_{i\in N} \bar x_i^{\,2}$ indeed minimizes the variance of $(\bar x_1, \dots, \bar x_n)$, as is well-known:
Since $\frac{1}{n} \sum_{i\in N} \bar x_i = \frac{k}{n}$, it holds that the variance of $(\bar x_1, \dots, \bar x_n)$ equals
\begin{align*}
\frac{1}{n}\sum_{i\in N} \left( \bar x_i -\frac{k}{n}\right)^2
&= \frac{1}{n}\sum_{i\in N} \left( \bar x_i^{\, 2}-2 \bar x_i\cdot \frac{k}{n} + \frac{k^2}{n^2}\right) \\
&= \frac{1}{n}\sum_{i\in N} \bar x_i^{\, 2} - \frac{1}{n} \cdot 2 k\cdot \frac{k}{n} + \frac{1}{n}\cdot n\cdot \frac{k^2}{n^2} \\
&= \frac{1}{n}\sum_{i\in N} \bar x_i^{\, 2} - \frac{k^2}{n^2}.
\end{align*}
When minimizing this expression, we can ignore multiplicative or additive constants ($n$ and $k$) and thus equivalently minimize $\sum_{i\in N} \bar x_i^{\,2}$.

The following example demonstrates that the maximum voter load under \varP may indeed be greater than under \maxP.

\begin{example}\label{ex:varmax}
	Let $C = \set{a,b,c,d}$, $k=3$, and consider the $5$-voter preference profile 
    given by $A_1=\{a\}$, $A_2=\{b\}$, $A_3=\{b,c\}$, $A_4=\{a,b,c\}$, $A_5=\{d\}$.
	For this instance, \maxP selects the committee $\{a,b,c\}$ and \varP selects the committee $\{a,b,d\}$. Optimal load distributions corresponding to these committees are illustrated in \figref{fig:varmax}.
	Load distributions minimizing the maximum voter load (like the one illustrated by the first diagram in \figref{fig:varmax}) satisfy $\max_{i\in N} \bar x_i = \frac{3}{4}$ and $\sum_{i\in N} \bar x_i^2 = 4 (\frac{3}{4})^2 = \frac{9}{4}$, and 
	the load distribution minimizing the variance of voter loads (illustrated by the second diagram in \figref{fig:varmax}) satisfies $\max_{i\in N} \bar x_i = 1$ and $\sum_{i\in N} \bar x_i^2 = 4 (\frac{1}{2})^2 + 1^2 = 2$. 
\end{example}

\newcommand{\clrstr}{20}
\begin{figure}[tb]
	\centering
\begin{tikzpicture}[yscale=.5,xscale=.5]
	\draw[fill=orange!\clrstr] (0,7) rectangle node {$a$} +(6,1); 
	\draw (0,6)[fill=violet!\clrstr] rectangle node {$b$} +(6,1);
	\draw[fill=teal!\clrstr] (0,5) rectangle node {$c$} +(6,1);
	\draw[fill=orange!\clrstr] (0,4) rectangle node {$a$} +(2,1); \draw[fill=violet!\clrstr] (2,4) rectangle node {$b$} +(2,1); \draw[fill=teal!\clrstr] (4,4) rectangle node {$c$} +(2,1);
	\draw (0,3) rectangle node {} +(0,1); 
	\draw[->] (0,3)--(8.5,3);
	\draw (0,3) -- (0,2.8)
	node[anchor=north] {$0$};
	\draw (2,3) -- (2,2.8)
	node[anchor=north] {$\frac{1}{4}$};
	\draw (4,3) -- (4,2.8)
	node[anchor=north] {$\frac{1}{2}$};
	\draw (6,3) -- (6,2.8)
	node[anchor=north] {$\frac{3}{4}$};
	\draw (8,3) -- (8,2.8)
	node[anchor=north] {$1$};
	
	\draw (-0.5,7.5) node[left]{$A_1 = \{a\}$};
	\draw (-0.5,6.5) node[left]{$A_2 = \{b \}$};
	\draw (-0.5,5.5) node[left]{$A_3 = \{b, c\}$};
	\draw (-0.5,4.5) node[left]{$A_4 = \{a, b, c\}$};
	\draw (-0.5,3.5) node[left]{$A_5 = \{d\}$};
\end{tikzpicture}
\hfill
\begin{tikzpicture}[yscale=.5,xscale=.5]
	\draw[fill=orange!\clrstr] (0,7) rectangle node {$a$} +(4,1); 
	\draw[fill=violet!\clrstr] (0,6) rectangle node {$b$} +(4,1);
	\draw[fill=violet!\clrstr] (0,5) rectangle node {$b$} +(4,1);
	\draw[fill=orange!\clrstr] (0,4) rectangle node {$a$} +(4,1); 
	\draw[fill=magenta!\clrstr] (0,3) rectangle node {$d$} +(8,1); 
	
	%axis
	\draw[->] (0,3)--(9,3);
	%ticks
	% 0
	\draw (0,3) -- (0,2.8)
	node[anchor=north] {$0$};
	% 1/2
	\draw (4,3) -- (4,2.8)
	node[anchor=north] {$\frac{1}{2}$};
	% 1
	\draw (8,3) -- (8,2.8)
	node[anchor=north] {$1$};
	
	\draw (-0.5,7.5) node[left]{$A_1 = \{a\}$};
	\draw (-0.5,6.5) node[left]{$A_2 = \{b \}$};
	\draw (-0.5,5.5) node[left]{$A_3 = \{b, c\}$};
	\draw (-0.5,4.5) node[left]{$A_4 = \{a, b, c\}$};
	\draw (-0.5,3.5) node[left]{$A_5 = \{d\}$};
\end{tikzpicture}
\caption{Illustration of Example \ref{ex:varmax}. The diagram on the left illustrates a load distribution minimizing the maximum voter load $\max_{i\in N} \bar x_i$, and the diagram on the right illustrates the unique load distribution minimizing $\sum_{i\in N} \bar x_i^2$.}
\label{fig:varmax}
\end{figure}
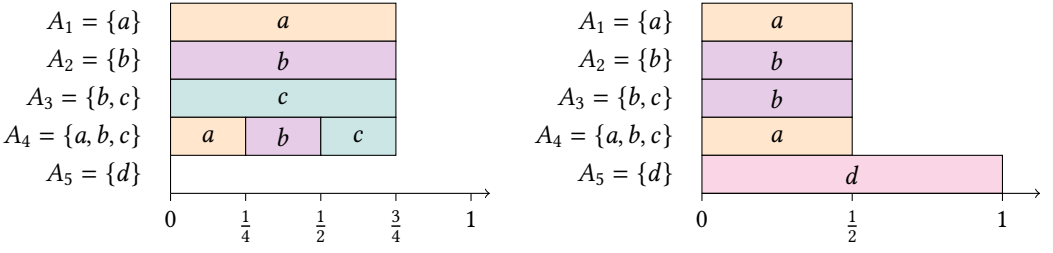

\begin{remark} \label{rem:minimaxP}
Rather than minimizing the maximum load, one could also aim to \emph{maximize the minimum voter load}. This variant  would select committees minimizing the number of unrepresented voters, even in the face of large cohesive groups of voters. Therefore, this method will not do well in terms of the representation axioms considered in \secref{sec:representation}. For this reason, we do not consider it further in this paper.
\end{remark}

\subsection{Sequential Method}
\label{sec:sequential}

We now introduce the sequential method, which can be seen as a greedy algorithm for minimizing the maximum voter load.  

\smallskip
\noindent \textbf{\seqP:}
The rule \seqP starts with an empty committee and
iteratively adds candidates, always choosing the candidate that minimizes
the (new) maximum voter load 
(under the assumption that previously assigned loads cannot be redistributed).
Let $\bar x_i^{(j)}$ denote the voter loads after round $j$.
At first, all voters have a load of $0$, i.e., $\bar x_i^{(0)}=0$ for all
$i\in N$. 
In each round, we keep the already assigned loads, but we may
further increase them and give the additional load to a new candidate~$c$.
In other words, we require $\bar x_i^{(j)}\ge \bar x_i^{(j-1)}$ for all $i$,
with equality unless $i\in N_c$. Moreover, 
the sum of the loads added in the round should be $1$.
(Hence, the total load after $j$ rounds is $j$, which is the sequential
version of constraint \eqref{eq:optimal-cond1}.)
We select the candidate $c$ and the loads $\bar x_i^{(j)}$
that satisfy these conditions and minimize $\max_i \bar x_i^{(j)}$.
(If there are several candidates achieving the minimum, we use a fixed
tie-breaking rule to decide which candidate to add.) 

\smallskip

The candidates and loads chosen by this procedure have the following
properties. 

\begin{lemma}\label{lem:seq}
In round $j$, 
given the voter loads $\bar x_i^{(j-1)}$ for all $i\in N$
and a candidate $c$ that was not selected in earlier rounds,
let
\begin{align}\label{seq1}
s_c^{(j)} = \frac{1 + \sum_{i\in N_c} \bar x_i^{(j-1)}}{|N_c|}.
\end{align}
Then, the maximum load $s^{(j)}=\max_i \bar x_i^{(j)}$
after round $j$ will be
\begin{align}\label{seq0}
  s^{(j)}=\min_c   s_c^{(j)},
\end{align}
taking the minimum over the candidates that remain in round $j$,
and a candidate $c$ is elected that achieves the minimum in \eqref{seq0}.
Moreover, if $c$ is elected, the new loads after round $j$ will be
\begin{align}\label{seq2}
\bar x_i^{(j)} = \begin{cases}
s_c^{(j)} & \text{if } i \in N_c\\
\bar x_i^{(j-1)} & \text{otherwise.}
\end{cases}
\end{align}

Furthermore, both individual loads and
the maximum load sequence are weakly increasing:
\[0\le \bar x_i^{(1)} \le \ldots \le \bar x_i^{(k)} \text{ for every $i \in N$, and }
0\le s^{(1)} \le \ldots \le s^{(k)}.\]
\end{lemma}

\begin{proof}
We use induction on $j$, so we assume that the claims hold for all
rounds before~$j$.
We claim first that the following inequalities hold
for every remaining candidate $c$ and for all $i\in N$:
\begin{align}  \label{seq04}
s_c^{(j)} \ge s^{(j-1)} \ge \bar x_i^{(j-1)}.
\end{align}
It is obvious that \eqref{seq04} holds for $j=1$.
If $j>1$, then, by the induction hypothesis,
$\bar x_i^{(j-1)}\ge \bar x_i^{(j-2)}$ for every $i$.  Hence, \eqref{seq1} 
yields
$s_c^{(j)} \ge s_c^{(j-1)} $ for every remaining candidate $c$.
Furthermore,
\eqref{seq0} (for $j-1$) yields
$s_c^{(j-1)}  \ge s^{(j-1)} $ for every remaining candidate $c$,
and thus \eqref{seq04} holds in this case too, 
recalling the definition $s^{(j-1)}=\max_i \bar x_i^{(j-1)}$.

Next, since
\eqref{seq04} holds,  for any remaining candidate~$c$, 
the assignment \eqref{seq2} satisfies
$\bar x_i^{(j)}\ge\bar x_i^{(j-1)}$ for every $i$, with equality if $i\notin N_c$.
Moreover, the sum of the added loads is, by \eqref{seq2} and \eqref{seq1},
\begin{align}\label{seq05}
  \sum_{i\in N_c}\left(\bar x_i^{(j)}-\bar x_i^{(j-1)}\right)
=|N_c| s_c^{(j)}-\sum_{i\in N_c}\bar x_i^{(j-1)}
=1.
\end{align}
Thus, \eqref{seq2} yields a valid load distribution for round $j$.
It follows from \eqref{seq04} that its maximum load is $s_c^{(j)}$.

Conversely, any distribution of an additional load 1 on the voters in $N_c$
will give these voters an average load of $s_c^{(j)}$, and thus the maximum
load will be at least $s_c^{(j)}$ (and strictly greater for loads differing
from \eqref{seq2}).

Hence, the maximum load after round $j$ is minimized by one of the
assignments \eqref{seq2}, where obviously $c$ should be chosen to minimize
$s_c^{(j)}$.
This proves \eqref{seq0}, and the remaining assertions follow.
\end{proof}

Note that \eqref{seq1}--\eqref{seq2} (together with a
tie-breaking rule) give a simple polynomial-time algorithm for computing the outcome of \seqP: 
In each round $j$, compute $s_c^{(j)}$ for all remaining candidates $c$, select a candidate minimizing this quantity (potentially using the tie-breaking rule), and update voter loads according to \eqref{seq2}. We analyze the running time of this algorithm in more detail in \secref{sec:computation}.

\smallskip

\citet{Phra99a} illustrates his sequential method by imagining the different ballots as represented by cylindrical vessels, with base area proportional to the number of voters casting that ballot.
The already elected candidates are represented by a liquid that is fixed in the vessels, and the additional unit of load incurred by adding another candidate to the committee is represented by pouring 1 unit of a liquid into the vessels representing voters approving this candidate.
The liquid then distributes among these vessels so that the height of the liquid is the same in all vessels. This is to be tried for each candidate; the candidate that requires the smallest height is elected, and the corresponding amounts of liquid are added to the vessels and fixed there.

An alternative interpretation of the sequential method is in terms of money: Imagine that voters have initially empty bank accounts and earn money continuously (at a constant rate) over time. As soon as the approvers of a candidate jointly own one dollar, they ``buy'' this candidate and their bank accounts are reset to zero.  This interpretation was utilized by \citet{peters2019proportionality} when introducing the Method of Equal Shares. 

\phrag's sequential method is committee monotonic by definition. 
As mentioned above, \seqP can be seen as a (polynomial-time computable) heuristic to approximate the \opti method \maxP. Unsurprisingly, the load distribution constructed by \seqP might not be optimally balanced.\footnote{The approximability of \maxP has recently been studied by \citet{CeSt21a}, who showed, in particular, that \seqP does not offer a constant-factor approximation guarantee.} 

\begin{example} \label{ex:seq}
Consider again the instance from \exref{ex:varmax}. In the first round, we have
$s_b^{(1)}=\frac{1}{3}$, $s_a^{(1)}=s_c^{(1)}=\frac{1}{2}$, and
$s_d^{(1)}=1$. Therefore, candidate $b$ is chosen. 
In the second round, we have $s_a^{(2)} = \frac{2}{3}$, $s_c^{(2)}=\frac{5}{6}$,
and $s_d^{(2)}=1$, so candidate $a$ is chosen. In the third round, there is a tie between~$c$ and $d$ because $s_c^{(3)}=s_d^{(3)}=1$. Thus, the final committee is either $\set{a,b,c}$ or $\set{a,b,d}$, depending on which tie-breaking rule is used. 
\figref{fig:seq} illustrates the resulting load distributions, both of which are suboptimal for the optimization problems corresponding to \maxP and \varP.
\end{example}

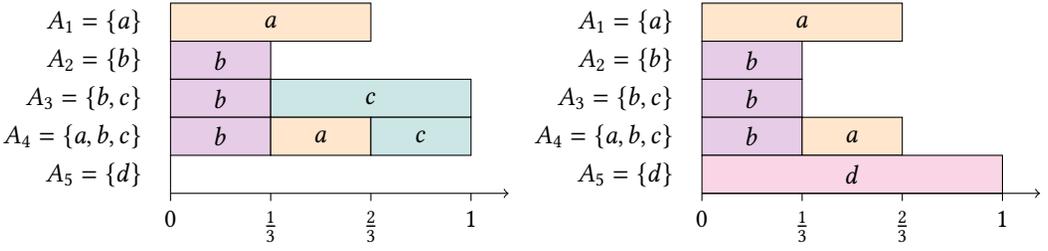
\begin{figure}[tb]
	\centering
\begin{tikzpicture}[yscale=.5,xscale=.5]
	\draw[fill=orange!\clrstr] (0,7) rectangle node {$a$} +(16/3,1); 
	\draw[fill=violet!\clrstr] (0,6) rectangle node {$b$} +(8/3,1);
	\draw[fill=violet!\clrstr] (0,5) rectangle node {$b$} +(8/3,1); \draw[fill=teal!\clrstr] (8/3,5) rectangle node {$c$} +(16/3,1); 
	\draw (0,4)[fill=violet!\clrstr] rectangle node {$b$} +(8/3,1); \draw[fill=orange!\clrstr] (8/3,4) rectangle node {$a$} +(8/3,1); \draw[fill=teal!\clrstr] (16/3,4) rectangle node {$c$} +(8/3,1);
	\draw (0,3) rectangle node {} +(0,1); 
	
	% 1/3
	\draw (8/3,3) -- (8/3,2.8)
	node[anchor=north] {$\frac{1}{3}$};
	% 2/3
	\draw (16/3,3) -- (16/3,2.8)
	node[anchor=north] {$\frac{2}{3}$};
	% 0
	\draw (0,3) -- (0,2.8)
	node[anchor=north] {$0$};
	% 1
	\draw (8,3) -- (8,2.8)
	node[anchor=north] {$1$};
	%axis
	\draw[->] (0,3)--(9,3);
	
	\draw (-0.5,7.5) node[left]{$A_1 = \{a\}$};
	\draw (-0.5,6.5) node[left]{$A_2 = \{b \}$};
	\draw (-0.5,5.5) node[left]{$A_3 = \{b, c\}$};
	\draw (-0.5,4.5) node[left]{$A_4 = \{a, b, c\}$};
	\draw (-0.5,3.5) node[left]{$A_5 = \{d\}$};
\end{tikzpicture}
\hfill
\begin{tikzpicture}[yscale=.5,xscale=.5]
	\draw[fill=orange!\clrstr] (0,7) rectangle node {$a$} +(16/3,1); 
	\draw[fill=violet!\clrstr] (0,6) rectangle node {$b$} +(8/3,1);
	\draw[fill=violet!\clrstr] (0,5) rectangle node {$b$} +(8/3,1);
	\draw[fill=violet!\clrstr] (0,4) rectangle node {$b$} +(8/3,1); \draw[fill=orange!\clrstr] (8/3,4) rectangle node {$a$} +(8/3,1); 
	\draw[fill=magenta!\clrstr] (0,3) rectangle node {$d$} +(8,1); 
	
	% 1/3
	\draw (8/3,3) -- (8/3,2.8)
	node[anchor=north] {$\frac{1}{3}$};
	% 2/3
	\draw (16/3,3) -- (16/3,2.8)
	node[anchor=north] {$\frac{2}{3}$};
	% 0
	\draw (0,3) -- (0,2.8)
	node[anchor=north] {$0$};
	% 1
	\draw (8,3) -- (8,2.8)
	node[anchor=north] {$1$};
	%axis
	\draw[->] (0,3)--(9,3);
	
	\draw (-0.5,7.5) node[left]{$A_1 = \{a\}$};
	\draw (-0.5,6.5) node[left]{$A_2 = \{b \}$};
	\draw (-0.5,5.5) node[left]{$A_3 = \{b, c\}$};
	\draw (-0.5,4.5) node[left]{$A_4 = \{a, b, c\}$};
	\draw (-0.5,3.5) node[left]{$A_5 = \{d\}$};
\end{tikzpicture}
\caption{Illustration of Example \ref{ex:seq}. The diagram on the left (respectively, right) illustrates the load distributions obtained by \seqP with ties broken in favor of candidate $c$ (respectively, $d$).}
\label{fig:seq}
\end{figure}

One can also define a sequential version of \varP, by in each iteration selecting a candidate minimizing the variance of the resulting load distribution~\citep{Mora16a}. 
This variant does not fare well in terms of the representation axioms considered in \secref{sec:representation}, and we therefore do not consider it any further.

\subsection{\eneP Method}
\label{sec:enestroem}

In addition to the methods described in the previous sections, there is another rule that is attributed, at least partially, to \phrag.\footnote{For details on the origin of this rule, see footnote 38 in the paper by \citet{Jans16a}, who refers to this method as \textit{\phrag's first method}. \citet{SEL17a} refer to the method as \textit{\phrag-STV}.}  Following \citet{camps2019method}, we refer to this method as \textit{\eneP}.

The method predates the load balancing methods and is similar in spirit to single transferable vote (STV) methods~\citep{Tide95a}. It uses a quota $q$, which is defined either as the \textit{Hare quota} $q_H = \frac{n}{k}$ or as the \textit{Droop quota} $q_D  = \frac{n}{k+1}$. The choice between $q_H$ and $q_D$ does not affect the axiomatic performance of the rule with respect to the properties studied in this paper (see Table \ref{tbl:jr2}).
While \eneP is indistinguishable from \seqP with respect to the representation properties studied in Section~\ref{sec:representation}, a crucial difference is that \eneP is not committee monotonic \cite{camps2019method}.

\smallskip
\noindent \textbf{\eneP:}
Initially, all voters have a voting weight of~$1$. Each ballot is counted fully, with its present voting weight, for each unelected candidate on the ballot. In each round, a candidate with maximum weighted approval score is chosen and the voting weights of voters approving this candidate are reduced: If the maximum weighted approval score $v$ is strictly greater than the quota (i.e., $v > q$), then each of these ballots has its voting power multiplied by $\frac{v-q}{v}$; %$(v − q)/v$; 
if $v \le q$, then these ballots all get voting power $0$ (and are thus ignored in the sequel). This is repeated until the desired number of candidates are elected.

\smallskip

Note that the total voting weight of all voters is decreased by $(q/v) \cdot v = q$ each time, as long as some candidate reaches the quota. 
This rule has been extensively analyzed by \citet{camps2019method} (mostly using $q_D$). Independently, it has been studied by \citet{SEL17a} (using $q_H$). In the following example, we use $q_H$. 

\begin{example} \label{ex:eneP}
Consider again the instance from \exref{ex:varmax}. We have $q_H=\frac{n}k=\frac{5}3$. 
In the first round, candidate $b$ is chosen with a (weighted) approval score of $3$. Since $3>q_H$, the voting power of the three voters approving $b$ is multiplied by $\frac{3-q_H}{3}=\frac{4}9$. 
In the second round, the weighted approval scores of the remaining candidates are $1+\frac{4}9 = \frac{13}9$ for $a$,  
$\frac{4}9 +\frac{4}9 = \frac{8}9$ for $c$, 
and $1$ for $d$. Therefore, candidate $a$ is chosen. 
Since $\frac{13}9 \le q_H$, both voters approving $a$ have their voting power reduced to $0$.
In the third and final round, the weighted approval score of $c$ is $\frac{4}9$ and candidate $d$ is chosen with a weighted approval score of $1$.
\end{example}

\section{Computational Aspects}
\label{sec:computation}

In this section, we study the computational complexity of \phrag's methods, and we provide algorithms for finding winning committees.
\citet{SFF+17a} have shown that every rule satisfying \emph{perfect representation} (see \secref{sec:representation}) is NP-hard to compute; this essentially follows from earlier work by \citet{PSZ08a}. 
Since we show that \maxP 
and \varP both satisfy this condition (Theorems~\ref{thm:maxP-PR} and~\ref{thm:varP-PR}), it follows that there do not exist polynomial-time algorithms for computing a committee for either of these rules, unless 
$\text{P}=\text{NP}$.

We complement these hardness results by considering two basic decision problems.  
\textsc{\maxP} asks whether an instance allows a load distribution~$x$ such that $(\bar x_{1},\dots, \bar x_{n}) \mathbin{\dot<} (y_1,\dots,y_n)$ for some given $n$-tuple $(y_1,\dots,y_n) \in \mathbb{R}_{\ge 0}^n$.
\textsc{\varP} asks whether an instance allows a load distribution $x$ such that $\sum_{i\in N} \bar x_{i}^{\,2} < \alpha$ for some given threshold value $\alpha > 0$. 
Both problems can be interpreted as asking whether a given load distribution is optimal.
We show that both problems are NP-complete even for rather restricted instances. 
For a preference profile $A$, let $s(A)$ denote the maximum number of candidates a voter approves, 
and let $d(A)$ denote the maximum number of voters that approve a candidate.

\begin{theorem}
The decision problems \textsc{\maxP} and \textsc{\varP} are NP-complete, even restricted to instances with $s(A)=2$ and $d(A)=3$.
\label{thm:phragmen-np}
\end{theorem}

\begin{proof}
To show hardness for both problems, we reduce from the NP-complete problem \textsc{Independent Set} on cubic graphs~\citep{GJS76a,GaJo79a}, which is defined as follows:
given a cubic graph $(V,E)$ (i.e., a graph such that every vertex has degree 3) and a positive integer~$k$, is there a set of vertices $S\subseteq V$ with $|S|=k$ such that $|e\cap S|\leq 1$ for all edges $e\in E$?
Let $E=(e_1,\dots,e_n)$.
We construct an instance of \textsc{\maxP} and \textsc{\varP} by
identifying candidates with vertices ($C=V$) and voters with edges, i.e., $A=(e_1,\dots,e_n)$.
It is easy to see that $s(A)=2$ and $d(A)=3$.
Without loss of generality we assume that $n\geq 3k$ because cubic graphs
with fewer than $3k$ edges cannot have an independent set of size $k$.\footnote{To see this, consider a cubic graph
with an independent set of size $k$. All $k$ vertices in the independent set have three outgoing edges and these $3k$ edges must all be distinct, since vertices in an independent set must not be connected via an edge.}

To prove that \textsc{\maxP} is NP-hard, we claim that $(V,E)$ has an independent set of size~$k$ if and only if there exists a load distribution $x$ with $(\bar x_{1},\dots, \bar x_{n}) \mathbin{\dot<} (y_1,\dots,y_n)$, where $(y_1,\dots,y_n)$ is the sequence containing $3k$ entries of $\frac 1 3 + \frac{1}{9k} $ followed by zeros.
If $S$ is an independent set, then $S$, viewed as a committee, contains candidates that are approved by disjoint sets of (three) voters. Hence, there are exactly $3k$ voters that bear a load of $\frac{1}{3}$; all others have load $0$.
Conversely, let $S$ be a committee such that $(\bar x_{1},\dots, \bar x_{n})
\mathbin{\dot<} (y_1,\dots,y_n)$. Since candidates are approved by three
voters, if there exists a voter with more than one approved candidate in
$S$, then the average load (and thus the maximum load) is at least
$\frac{k}{3k-1}>\frac{1}{3}+\frac{1}{9k}$, which contradicts our
assumption that $(\bar x_{1},\dots, \bar x_{n}) \mathbin{\dot<}
(y_1,\dots,y_n)$. Hence $S$ is an independent set. 

To prove that \textsc{\varP} is NP-hard, we claim that $(V,E)$ has an independent set of size $k$ if and only if there exists a load distribution $x$ with $\sum_{i\in N} \bar x_{i}^{\,2}< \frac{k}{3} + \frac{1}{9}$. It is straightforward to see that an independent set $S$ corresponds to a committee with $\sum_{i\in N} \bar x_{i}^{\,2}= 3k\cdot \left(\frac{1}{3}\right)^2=\frac{k}{3}$.
For the other direction, let $S$ be a committee with $\sum_{i\in N} \bar
x_{i}^{\,2}< \frac{k}{3}+\frac{1}{9}$. Note that at most $3k$ voters have
approved candidates in the committee. Let $N'\subseteq N$ be such that it
contains all voters $i$ with $\bar x_i>0$. Hence $\sum_{i\in N} \bar
x_{i}^{\,2}=\sum_{i\in N'} \bar x_{i}^{\,2}$. The value of $\sum_{i\in N'}
\bar x_{i}^{\,2}$ is minimal only if all $\bar x_i$, $i\in N'$, are equal and we
then have
$\sum_{i\in N'} \bar x_{i}^{\,2} = |N'|\cdot (\frac{k}{|N'|})^2 =\frac{k^2}{|N'|}$ \text.
If $|N'|<3k$, we thus see that $\sum_{i\in N'} \bar x_{i}^{\,2} \geq \frac{k^2}{3k-1} > \frac{k}{3}+\frac{1}{9}$.
Hence $|N'|=3k$ and we can conclude that $S$ corresponds to an independent set.

It remains to be shown that \textsc{\maxP} and \textsc{\varP} are contained in NP.
This is not immediate as a witness for a Yes-Instance (i.e., a load distribution) may not have a polynomially-sized bit representation. In other words, the fractions in the load distribution may have very large numerators and denominators.
To resolve this issue, we encode \textsc{\maxP} as a mixed-integer linear program (see the discussion following this proof).
Solving a mixed-integer linear program (i.e., its corresponding decision problem) is known to be NP-complete~\citep{Schr86a}.\footnote{This result essentially shows that mixed-integer linear programs have solutions of polynomial size.}
For showing NP-membership of \textsc{\varP}, we proceed in a similar fashion: we encode it as a mixed-integer quadratic program (see Theorem~\ref{thm:varP-mqlp}). NP-membership  then follows from a result by~\citet{PDM17a}.
 \end{proof}

We now turn to algorithms for computing \phrag's methods. First, we show how the outcome of \maxP can be computed with the help of mixed-integer linear programs (MILPs).\footnote{
For a general discussion on lexicographic optimization in MILPs, we refer the reader to a paper by \citet{DBLP:conf/iccsa/OgryczakS06} and references therein.}
We start by formulating a MILP that solves the decision problem \textsc{\maxP}.
We are thus given a load vector $\mathbf{y} = (y_1, \ldots, y_n)$ and ask whether an improvement is possible.
Without loss of generality we assume that $y_1 \geq  \ldots \geq y_n$.
The general idea is to find an index $t$ where an improvement over $\mathbf{y} = (y_1, \ldots, y_n)$ is possible.
This requires a new load vector $\mathbf{x} = (\bar x_1, \ldots, \bar x_n)$
such that $\bar x_{(1)},\dots,\bar x_{(t-1)}$ remain equal to $y_1, \ldots, y_{t-1}$, respectively, and that $\bar x_{(t)},\dots,\bar x_{(n)}$ are each less than or equal to $y_t - \epsilon$ for some $\epsilon>0$.
We thus guess the index~$t$ and a mapping from indices $1,\dots,t-1$ to voters.

We use variables $x_{i,c}$ (for $i\in N$, \mbox{$c\in C$}), $e_{i,j}$ (for $i,j\in N$), $s_{i}$ (for $i\in N$), $t_{j}$ (for $j\in N$), and $\epsilon$.
Recall that $\bar x_i=\sum_{c\in C} x_{i,c}$.
For a given $n$-tuple $\mathbf{y}$, let $\milpmax(\mathbf{y})$ be the MILP that maximizes $\epsilon$ under the constraints \eqref{eq:optimal-cond3a}--\eqref{eq:optimal-cond2} and \eqref{ilp:pi1}--\eqref{ilp:sm5}.
{\allowdisplaybreaks
\begin{align}
&e_{i,j} \in \set{0,1} &\quad \text{ for all $i,j \in N$} \label{ilp:pi1}\\
&s_i \in \set{0,1} & \text{ for all $i \in N$} \label{ilp:sm1}\\
&t_j \in \set{0,1} & \text{ for all $j \in N$} \label{ilp:sm2}\\
&s_i+\sum_{j \in N} e_{i,j}  = 1 & \text{ for all $i \in N$} \label{ilp:pi3}\\
&t_j+\sum_{i \in N} e_{i,j}  \leq 1 & \text{ for all $j \in N$} \label{ilp:pi2}\\
&\sum_{j \in N} t_j = 1 &  \label{ilp:sm3}\\
&\bar x_i - k (1 - e_{i,j}) \le y_j & \text{ for all $i,j \in N$} \label{ilp:xy}\\
&\bar x_i - k (2 - s_i - t_j) \le y_j -\epsilon & \text{ for all $i,j \in N$} \label{ilp:sm5}
\end{align}}

This MILP can be understood as follows:
The variables $e_{i,j}$ encode a partial bijection $\pi$ from a subset of $N$ to a subset of $N$ (those indices where no improvement occurs); the variables $s_i$ encode the subset $S \subseteq N$ where $\pi$ is not defined (those indices where the loads are less than or equal to $y_t-\epsilon$);
and the variables $t_j$ encode $t\in N$, an index of an element in $\{y_j: j\notin\textit{range}(\pi)\}$ (the index~$t$ where an actual improvement occurs).
Constraint \eqref{ilp:pi3} encodes the relation between~$\pi$ and~$S$:
for every $i\in N$, either $s_i=1$ or $e_{i,j}=1$ for some $j\in N$.
In a similar fashion, constraint \eqref{ilp:pi2} encodes the relation between $\pi$ and $t$:
for every $i\in N$, $t_i=1$ only if $e_{i,j}=0$ for all $j\in N$.
Together with constraint~\eqref{ilp:sm3}, we enforce that there exists exactly one $j\in N$ such that $t_j=1$.
Hence at least one voter has a load strictly smaller than $y_t$ and $(\bar x_{1},\dots, \bar x_{n}) \mathbin{\dot<} (y_1,\dots,y_n)$.

The final two constraints ensure that indeed $(\bar x_{1},\dots, \bar x_{n}) \mathbin{\dot<} (y_1,\dots,y_n)$.
From constraint~\eqref{ilp:xy} it follows that $\bar x_i \leq y_j$ whenever $\pi(i)=j$.
This is because if $e_{i,j}=0$ (i.e., $\pi(i)\neq j$), constraint~\eqref{ilp:xy} reduces to $\bar x_i - k \le y_j$, which is trivially satisfied because every load distribution $x$ satisfies $\bar x_i \le k$ for all $i \in N$.
If $e_{i,j}=1$ (i.e., $\pi(i)= j$), however, constraint \eqref{ilp:xy} reads $\bar x_i \le y_j$.
Similarly, constraint~\eqref{ilp:sm5} enforces that $\bar x_i\leq y_t -\epsilon \leq \max_{j\in N\setminus \textit{range}(\pi)} y_j -\epsilon$ for $i\in S$.
As we maximize $\epsilon$, we look for a solution where $\bar x_i < \max_{j\in N\setminus \textit{range}(\pi)} y_j$.
We conclude that a feasible solution with objective function value $\epsilon>0$ encodes a load distribution $x$ with $(\bar x_{1},\dots, \bar x_{n})
\mathbin{\dot<}(y_1,\dots,y_n)$. 
Observe that $\milpmax(\mathbf{y})$ solves the \textsc{\maxP} decision problem:
given voter loads $\mathbf{y}$, $\milpmax(\mathbf{y})$ returns $\epsilon>0$ if and only if \textsc{\maxP} with input $\mathbf{y}$ is a Yes-instance.

\smallskip
We now present a MILP-based algorithm that computes the outcome of \maxP. Our algorithm solves a sequence of at most $2n$ instantiations of the MILP~$\milpmax$, using the optimal solutions of previously solved instances as constraints for subsequent calls.
We assume that $\milpmax$ returns the load distribution $x$ and the objective function value $\epsilon$.
For an overview of the procedure, see Algorithm~\ref{alg:maxP}.

{
\renewcommand{\baselinestretch}{1.2}
\begin{algorithm}
    \normalsize
	\DontPrintSemicolon
	$\mathbf{y}\gets(k,0,\dots,0)$\;
	\For{$\ell=1\dots n$}{
		$x,\epsilon\gets\milpmax(\mathbf{y})$ \;
		$\bar x\gets(\bar x_{1},\dots, \bar x_{n})$\tcp*{$\bar x_{(1)} ,\dots, \bxell$ optimal}				
		\If(\tcp*[f]{no improvement}){$\epsilon = 0$}{
			$x',\epsilon'\gets\milpmax({\bar x})$\;
			\If(\tcp*[f]{$\bar x$ optimal}){$\epsilon'=0$}{
				\Return \ $\{c \in C \midd \sum_{i \in N} x_{i,c} = 1\}$\;
			}
		}		
		{	
			$\mathbf{y}\gets (\bar x_{(1)},\dots,\bar x_{(\ell+1)},0,\dots,0)$\;
		}
	}
	\Return $\{c \in C \midd \sum_{i \in N} x_{i,c} = 1\}$\;

	\caption{Computing \maxP}	\label{alg:maxP}
\end{algorithm}
}

We start with $\mathbf{y}=(k,0,\dots,0)$, an $n$-tuple consisting of one $k$ and $n-1$ zeros.
We employ $\milpmax$ to find a strictly better solution.
The only entry of $\mathbf{y}$ that can be improved is $\mathbf{y}_{(1)}=k$ and hence the solution $x$ returned by $\milpmax$ minimizes the largest load; let $\bar{x}_{(1)}$ be the largest load and $\bar{x}_{(2)}$ the second-largest.
We repeat this procedure with $\mathbf{y}=(\bar{x}_{(1)},\bar{x}_{(2)},0,\dots,0)$.
We already know that $\bar{x}_{(1)}$ is optimal and cannot be further decreased (and 0 cannot be improved), hence the next $\milpmax$ instance minimizes the second-largest load.
We iterate this process and in step $\ell$ guarantee that the $\ell$-th largest load is optimal.
If at some point $\milpmax$ returns $\epsilon=0$, we verify whether the current solution is optimal:
if $\milpmax({\bar x})$ also returns $\epsilon=0$, the load distribution $x$ is indeed optimal and
the algorithm terminates.
In any case Algorithm~\ref{alg:maxP} returns $\{c \in C \midd \sum_{i \in N} x_{i,c} = 1\}$, the committee corresponding to the load distribution $x$.

\smallskip

We have therefore proven the following result.

\begin{theorem}
\maxP can be computed by solving at most $2n$ mixed-integer linear programs with $\mathcal{O}(nm+n^2)$ variables.
\end{theorem}

To compute \varP, we solve a mixed-integer quadratic program (MIQP), i.e., a program consisting of linear constraints and a quadratic optimization statement.

\begin{theorem}\label{thm:varP-mqlp}
\varP can be computed by solving one mixed-integer quadratic program with $\mathcal{O}(n m)$ variables.
\end{theorem}

\begin{proof}
Our MIQP uses the variables $x_{i,c}$ (for $i\in N$, $c\in C$) 
and the constraints \eqref{eq:optimal-cond3a}--\eqref{eq:optimal-cond2}. 
The quadratic optimization statement is 
\[\min \sum_{i\in N} \left(\sum_{c\in C} x_{i,c}\right)^2 \text.\]
Since minimizing  $\sum_{i\in N} \bar x_i^{\,2}$ minimizes the variance (see \secref{sec:direct}), this MIQP computes load distributions corresponding to \varP committees.
\end{proof}

Finally, we study the runtime for computing \seqP. A naive estimate is that \seqP can be computed in $\mathcal O(kmn)$ time. This estimate ignores the cost of computing the quantities $s_c^{(j)}$, \ie numerical operations are assumed to require constant time.
While this is a sensible assumption in many cases, here it is questionable since computing $s_c^{(j)}$ exactly  requires fractions with large numerators and denominators.
Indeed, the denominator of $s_c^{(j)}$ can grow exponentially with $j$.
Hence, the following theorem also takes the complexity of these operations into account.

\begin{theorem}\label{thm:seqP-runtime}
The output of \seqP can be computed in $\mathcal{O}(k^3mn(\log n)^2)$ time.
\end{theorem}

\begin{proof} 
In the following analysis we also consider the complexity of arithmetic operations in the algorithms, as exact numerical computation of the involved quantities may require numbers of substantial size.
Let us consider the procedure described in Section~\ref{sec:sequential}. In each of the $k$ rounds, one candidate is chosen. For this, the quantity $s_c^{(j)}$ is computed for every $c$ not yet placed in the committee.
To ensure correct results, we represent $s_c^{(j)}$ as fractions, i.e., pairs of integers.
Let $\{c_1,\dots,c_{j-1}\}$ be the first $j-1$ chosen candidates.
It is easy to see that the denominator of $s_c^{(j)}$ can be bounded by $|N_{c_1}|\cdot\ldots\cdot |N_{c_{j-1}}|\cdot |N_c|\leq n^j\leq n^k$, assuming we reduce fractions.
Furthermore, since $s_c^{(j)}\leq k$, the numerator of $s_c^{(j)}$ is at most $kn^k$.
Hence, the space required to store $s_c^{(j)}$ is bounded by $\mathcal{O}(k\log n)$.
The necessary computations for calculating $s_c^{(j)}$ (addition, division,
reducing fractions) can all be performed in $\mathcal O(b^2)$
time,\footnote{This quadratic bound is a very rough estimate and does not use
  any of the more sophisticated methods for multiplication such as the
  Sch\"onhage--Strassen algorithm~\citep{schonhage1971schnelle} or computing
  greatest common divisors~\citep{moller2008schonhage,brent2010modern}.}
where $b$ is the number of bits required to store any of $s_c^{(j-1)}$, and
$\mathcal{O}(n)$ such operations are required. 
Since $b= \mathcal{O}(k\log n)$, we conclude that $s_c^{(j)}$ can be computed in $\mathcal{O}(nk^2(\log n)^2)$ time.
This has to be done in each of the $k$ rounds for at most $|C|=m$ many candidates $c\in C$.
The consequent update of $\bar x_i^{(j)}$ does not increase the runtime bound further.
\end{proof}

\section{\phrag's Methods and Representation}
\label{sec:representation}

In this section, we study which representation axioms are satisfied by \phrag's methods. Our results are summarized in Table \ref{tbl:jr2}. Particularly noteworthy are the results that \seqP satisfies PJR and that \maxP and \varP satisfy PR.
For completeness, the table also contains results obtained by \citet{SEL17a} and \citet{camps2019method} regarding \eneP.

\subsection{Representation Axioms}

We start by stating the definitions of \citet{justifiedRepresentation} and \citet{SFF+17a}.

\begin{definition}\label{def:jr}
	A committee $S \subseteq C$ with $|S|=k$ provides
	\begin{itemize}
		\item \emph{justified representation (JR)} if there does not exist a set $N^* \subseteq N$ of voters with $|N^*| \ge \frac{n}{k}$, $|\bigcap_{i \in N^*} A_i| \ge 1$ and $|S \cap A_i| = 0$ for all $i \in N^*$.
		\item \emph{proportional justified representation (PJR)} if there does not exist an integer \mbox{$\ell>0$} and a set $N^* \subseteq N$ of voters with $|N^*| \ge \ell \frac{n}{k}$, $|\bigcap_{i \in N^*} A_i| \ge \ell$ and $|S \cap (\bigcup_{i \in N^*} A_i)| < \ell$.
		\item \emph{extended justified representation (EJR)} if there does not exist an integer $\ell>0$ and a set $N^* \subseteq N$ of voters with $|N^*| \ge \ell \frac{n}{k}$, $|\bigcap_{i \in N^*} A_i| \ge \ell$ and $|S \cap A_i| < \ell$ for all $i \in N^*$.
	\end{itemize} 
A rule $f$ \emph{satisfies} JR (respectively, PJR or EJR) if, for every instance $(A, k)$, every committee $S \in f(A, k)$ provides JR (respectively, PJR or EJR).  
\end{definition}

It follows immediately from the definitions that a rule satisfying EJR also satisfies PJR, and that a rule satisfying PJR also satisfies~JR.\footnote{\citet{justifiedRepresentation} have introduced an additional proportionality axiom known as \textit{core stability}. Since core stability is more demanding than EJR, the rules considered in this paper do not satisfy core stability.}

\newcommand{\pr}{\mathit{PR}}
\newcommand{\mr}{\mathit{MR}}

The following definition is due to \citet{SFF+17a}.

\begin{definition}\label{def:pr1}
	Consider an instance $(A, k)$ such that $k$ divides $n=|N|$. A committee $S = \{c_1, \ldots, c_k\} \subseteq C$ provides \emph{perfect representation} if there exists a partition of the set $N$ of voters into $k$ pairwise disjoint subsets $N_1, \ldots, N_k$ such that, 
	for all $j \in \{1,\ldots,k\}$,
	$|N_j| = \frac{n}{k}$ and $c_j \in  \bigcap_{i \in N_j} A_i$. 
	Let $\mathit{PR}(A, k)$ denote the set of all committees providing perfect representation for the instance $(A, k)$.
	A rule $f$ satisfies \emph{perfect representation (PR)} if, for every instance $(A, k)$ where $k$ divides $n$ and $\pr(A, k) \neq \emptyset$, we have $f(A,k) \subseteq \pr(A,k)$.
\end{definition}

The following example, which also appears in the papers by \citet{justifiedRepresentation} and \citet{SFF+17a}, illustrates the requirements of the different axioms.

\begin{example}\label{ex:4-2}
Let $C= \{a,b,c,d,e,f\}$ and consider the $8$-voter preference profile given by  
$A_1=\{a\}$, 
$A_2=\{b\}$, 
$A_3=\{c\}$, 
$A_4=\{d\}$, 
$A_5=\{a,e,f\}$, 
$A_6=\{b,e,f\}$, 
$A_7=\{c,e,f\}$,
$A_8=\{d,e,f\}$. 
Let $k=4$ and assume that ties are broken alphabetically. 
Then, \seqP chooses $e$, $f$, $a$, and $b$ (in this order). 
The final loads are $(\bar x_1, \ldots, \bar x_8) = (\frac{3}{4},\frac{3}{4},0,0,\frac{3}{4},\frac{3}{4},\frac{1}{2},\frac{1}{2})$.
This is indeed not optimal as there is a perfect load distribution $y$ with $\bar y_i = \frac{1}{2}$ for all $i \in N$. The corresponding committee $\{a,b,c,d\}$ is selected by both \maxP and \varP. %See \figref{fig:ex4-2} for an illustration.

Let $\ell=2$ and consider the voter group $N^* = \set{5, 6, 7, 8}$ of size $\ell
\frac{n}{k} = 2 \frac{8}{4}=4$.
Since the voters in $N^*$ all approve candidates $e$ and $f$, a set of size $\ell = 2$, the conditions for JR, PJR, and EJR all bind. JR requires that at least one candidate approved by at least one voter in~$N^*$ is chosen. PJR requires that at least $2$ candidates are chosen that are each supported by at least one voter from~$N^*$, while EJR requires that some voter from~$N^*$ is represented twice. Thus, EJR dictates that either $e$ or $f$ is chosen. 
On the other hand, the only committee providing PR is $\{a,b,c,d\}$.
As a consequence, no rule can satisfy both PR and EJR.\footnote{The incompatibility of PR and EJR was first observed by \citet{SFF+17a}.}
Note that \maxP and \varP both violate EJR in this example, and that \seqP violates PR.
\eneP also yields $\{e,f,a,b\}$, and thus violates PR.
\end{example}

\newcommand{\yes}{\textcolor{green!50!black}{\ding{52}}\xspace}
\newcommand{\no}{\textcolor{red!70!black}{\ding{55}}\xspace}

\begin{table}

\begin{center}
\scalebox{1}{
\begin{tabular}{lllll}
\toprule
		 & \quad\ JR             				& \quad\  PJR        							& \quad\ EJR         				& \quad\ PR     				\\
\midrule
\seqP	 & \yes (Corollary \ref{cor:seqP-JR})  & \yes (Theorem \ref{thm:seqP-PJR})	&  \no (Example \ref{ex:seqP-EJR}) & \no \;(Example \ref{ex:4-2})  \\

\maxP	 & \yes (Corollary \ref{cor:maxP-JR}) 	& \yes (Theorem \ref{thm:lexP-PJR})	   								&  \no (Example \ref{ex:4-2})		& \yes (Theorem \ref{thm:maxP-PR}) \\
\varP	 & \yes (Theorem \ref{thm:var-jr})			   				& \no \;(Example \ref{ex:varP-PJR}) &  \no (Example \ref{ex:4-2})        & \yes (Theorem \ref{thm:varP-PR}) \\
\eneP     & \yes~\citep{SEL17a,camps2019method}   & \yes~\citep{SEL17a,camps2019method}	&  \no~\citep{SEL17a,camps2019method} & \no \;(Example \ref{ex:4-2})  \\
\bottomrule
\end{tabular}
}
\end{center}
\caption{\phrag's methods and representation axioms}
\label{tbl:jr2}
\end{table}

\subsection{Results for \seqP}

\newcommand{\maxload}{max-load\xspace}

In this section we establish our main result: \seqP satisfies proportional justified representation. 
We use the following notation.
For the committee~$S$ that is selected by \seqP (using a fixed tie-breaking rule), we can relabel the candidates so that $S=\set{c_1 \ldots, c_k}$ and candidate $c_j$ was chosen in round $j$. Then, we have
$c_j = \arg\min_{c \in C \setminus \set{c_1, \ldots, c_{j-1}}} s_c^{(j)}$,
and the maximum load after round $j$ is $s^{(j)} = s_{c_j}^{(j)}$. 
The following lemma formalizes the intuitively obvious fact
that, when computing the optimal distribution of the load of a candidate~$c$
among its voters, it never helps to restrict attention to a subset 
${N' \subset N_c}$. 

\begin{lemma}\label{lem:mon+subset}
	Fix an instance $(A,k)$.
 For $j \le k$, a candidate $c \in C$ that has not been elected before round~$j$, and a nonempty subset $N' \subseteq N_c$, 
let, as a generalization of \eqref{seq1},
\begin{align}\label{seq1'}
s_c^{(j)}[N'] = \frac{1 + \sum_{i\in N'} \bar x_i^{(j-1)}}{|N'|}.
\end{align}
Then 
$s_c^{(j)}[N']$ is the maximum voter load after optimally distributing 
an additional load of $1$ among all voters in $N'$, on top of the loads $\bar
x_i^{(j-1)}$. 
In particular, 
$s_c^{(j)} = s_c^{(j)}[N_c] \le s_c^{(j)}[N']$ for all $N'\subseteq N_c$. 
\label{lem2}
\end{lemma}

\begin{proof} 
That $s_c^{(j)}[N']$ is the maximum voter load after optimally distributing 
an additional load 1 among $N'$ follows by \lemref{lem:seq} (or its proof)
by replacing $N_c$ by $N'$; the only non-obvious part is that
$s_c^{(j)}[N']\ge \bar x_i^{(j-1)}$ for all $i\in N'$.

Since the optimal distribution of the addional load among  $N'$ is a
possible distribution among the larger set $N_c$, it is obvious that
the optimal distribution among $N_c$ is at least as good, and thus
$s_c^{(j)}[N_c] \le s_c^{(j)}[N']$.
\end{proof}

We are now ready to prove our main theorem.

\begin{theorem}\label{thm:seqP-PJR}
	\seqP satisfies PJR.
\end{theorem}

\begin{proof}
	PJR requires that $|S \cap (\bigcup_{i \in N^*} A_i)| \ge \ell$ for all groups $N^* \subseteq N$ of voters satisfying $|N^*| \ge \ell \frac{n}{k}$ and $|\bigcap_{i \in N^*} A_i| \ge \ell$ for some integer $\ell>0$. We show that \seqP satisfies a strictly stronger property by weakening the constraint $|N^*| \ge \ell \frac{n}{k}$ to $|N^*| > \ell \frac{n}{k+1}$. 
	
	Consider an instance $(A,k)$ and let $S$ be the committee selected by \seqP. Assume for contradiction that there exists a voter group $N^* \subseteq N$ and an integer $\ell>0$ with $|N^*| > \ell \frac{n}{k+1}$ such that $|\bigcap_{i \in N^*} A_i| \ge \ell$ and $|S \cap (\bigcup_{i \in N^*} A_i)| \le \ell-1$.
	
	Let $c \in (\bigcap_{i \in N^*} A_i) \setminus S$ and consider round $k$ (the last round) of the \seqP procedure. Adding candidate $c$ to the committee would have caused a maximum voter load of 

	\begin{align}\label{seq31}
		s_c^{(k)} &=   \frac{1 + \sum_{i\in N_c} \bar x_i^{(k-1)}}{|N_c|} 
				  \le \frac{1 + \sum_{i\in N^*} \bar x_i^{(k-1)}}{|N^*|} \nonumber \\
				  &\le \frac{1+(\ell-1)}{|N^*|} = \frac{\ell}{|N^*|} < \frac{k+1}{n}.
	\end{align}  
	Here, the first inequality follows from 
    \lemref{lem:mon+subset} (observe that $N^* \subseteq N_c$), the second
    inequality follows from $|S \cap (\bigcup_{i \in N^*} A_i)| \le \ell-1$,
    and the strict inequality follows from $|N^*| > \ell \frac{n}{k+1}$. 
	
	Let $c_k$ be the candidate that was chosen in round $k$. Since candidate
    $c$ was \emph{not} chosen, we have $c \neq c_k$ and $s_{c_k}^{(k)} \le
    s_c^{(k)}$. Using \lemref{lem:seq} and \eqref{seq31},
we have 
$s^{(k)} = s_{c_k}^{(k)} \le s_c^{(k)} <\frac{k+1}{n}$. 
In particular, this implies that at the end of round
    $k$, every voter $i \in N$ has a load $\bar x_i^{(k)}$ that is strictly
    less than $\frac{k+1}{n}$. Summing the loads over all voters, we get 
	\begin{align*}
		\sum_{i \in N} \bar x_i^{(k)} &= \sum_{i \in N^*} \bar x_i^{(k)} + \sum_{i \in N \setminus N^*} \bar x_i^{(k)}\\
		                         &\le (\ell - 1) + |N \setminus N^*| \cdot s^{(k)}\\
								 &< \ell - 1     + \frac{n}{k+1} (k+1-\ell) \frac{k+1}{n}
								 = k \text,
	\end{align*}
	where we have used the fact that $|N \setminus N^*| \le \frac{n}{k+1} (k+1-\ell)$.
	But $\sum_{i \in N} \bar x_i^{(k)} < k$ is a contradiction, because the sum of all voter loads (at the end of the \seqP procedure) must equal $k$. This completes the proof.
\end{proof}

\begin{remark}
We note that the proof of \thmref{thm:seqP-PJR} shows that \seqP satisfies a property that is strictly stronger than PJR, because the constraint on the size of the group~$N^*$ has been relaxed.\footnote{Replacing the constraint $|N^*| \ge \ell \frac{n}{k}$ with $|N^*| > \ell \frac{n}{k+1}$ is similar to replacing the Hare quota with the Droop quota in the context of single transferable vote elections (see \secref{sec:enestroem}). 
The condition $|N^*| > \ell \frac{n}{k+1}$ is the best possible here; see the paper by \citet{janson2018thresholds}.}
\end{remark}

\begin{remark}\label{rem:priceable}
In fact, in recent work \citet{peters2019proportionality} have shown that \seqP satisfies a stronger property that they call priceability. This in turn implies that \seqP satisfies Inclusion Proportionality for Solid Coalitions (IPSC)~\citep{aziz2021proportionally}, a property that lies between priceability and PJR.\footnote{We thank Jannik Peters for pointing out to us that the proof of \citet{peters2019proportionality} showing that priceability implies PJR can be easily adapted to show that priceability implies IPSC.}
\end{remark}

An immediate corollary of \thmref{thm:seqP-PJR} is that \seqP satisfies~JR.

\begin{corollary} \label{cor:seqP-JR}
	\seqP satisfies JR.
\end{corollary}

However, \seqP violates EJR, as the following example demonstrates.

\begin{example}\label{ex:seqP-EJR}
Let $C = \{a, b, c_1,c_2, \ldots, c_{12}\}$, $k=12$, and 
consider the following profile with $n=24$ voters: 
	\begin{align*}
		&2 \times \{ a,b,c_1\} & & 6 \times \{ c_1, c_2, \ldots, c_{12}\} \\
		&2 \times \{ a,b,c_2\} & & 5 \times \{ c_2, c_3, \ldots, c_{12}\} \\
		&					   & & 9 \times \{ c_3, c_4, \ldots, c_{12}\}
	\end{align*}
	\seqP selects $S = \{ c_1, c_2, \ldots, c_{12}\}$. %\footnote{See Table~\ref{tbl:calc} in the supplemental material for details.} 
	(For details of the calculation, see Table~\ref{tbl:calc} in the appendix.)
	To see that~$S$ does not provide EJR, consider the group $N^*$ consisting of the four voters on the left. We have $|N^*| = 4 = 2 \frac{n}{k}$ and $|\bigcap_{i \in N^*} A_i| = |\{a,b\}| = 2$. Therefore, EJR requires that at least one voter in $N^*$ approves at least $2$ candidates in $S$, which is not the case. 
\end{example}

Note that \seqP also fails PR (see \exref{ex:4-2}). 
This is not surprising, considering that PR is computationally intractable~\citep{SFF+17a}. 

\subsection{Results for \maxP}
\label{sec:max-results}

In \exref{ex:4-2}, \maxP selects the committee providing perfect representation.
We now show that \maxP satisfies PR in general.

\begin{theorem}\label{thm:maxP-PR}
	\maxP satisfies PR. 
\end{theorem}

\begin{proof} 
	Consider an instance $(A, k)$ and assume that $\pr(A, k) \neq \emptyset$ (otherwise, there is nothing to show). 
	Recall  that a load distribution $x = (x_{i,c})_{i \in N, c \in C}$ is perfect if $\bar x_i = \frac{k}{n}$ for all $i \in N$.
	We first show that there is a perfect load distribution. Let $\{c_1, \ldots, c_k\} \subseteq C$ be a committee providing perfect representation and let $N_1, \ldots, N_k$ be a corresponding partition of $N$. Define load distribution $x^*$ by
	\[ x^*_{i,c_j} = 	\begin{cases}
					\frac{k}{n} &\text{if $i \in N_j$,} \\
					0           &\text{otherwise.}
					\end{cases}
	\]
	It is straightforward to check that $x^*$ is a valid load distribution and that $x^*$ is perfect. 
	
	Clearly, a perfect load distribution is an optimal solution for the minimization problem in \maxP. It follows that \emph{every} optimal load distribution is perfect. 
	We now show that every perfect load distribution corresponds to a committee providing perfect representation. It then follows that every committee $S$ output by \maxP provides perfect representation for $(A, k)$.
 	
	Let $x = (x_{i,c})_{i \in N, c \in C}$ be a perfect load distribution and let $S$ be the corresponding committee, i.e., $S = \{c \in C: \sum_{i \in N} x_{i,c} = 1\}$.
    Define $M$ to be an $n \times n$ matrix with rows corresponding to voters and, for each $c \in S$, $\frac{n}{k}$ columns $c^1, c^2, \ldots c^{\frac{n}{k}}$ corresponding to candidate $c$. 
    For $i \in N$ and $c \in S$, define the entry of $M$ in row $i$ and column $c^j$ (for all $1 \le j \le \frac{n}{k}$) to be $x_{i,c}$.
	Every row of $M$ sums to $\sum_{c \in S} x_{i,c} \frac{n}{k} = \frac{n}{k} \bar x_i = 1$, and every column of $M$ sums to $\sum_{i \in N} x_{i,c} = 1$, so $M$ is doubly stochastic. We can now apply the Birkhoff--von Neumann
theorem and get that $M$ is a convex combination of permutation matrices. 
Choose a permutation matrix $P$ in this convex combination. 
$P$ encodes a bijection between the sets $N$ and $\bigcup_{c \in S}
\bigcup_{j=1}^{n/k} c^j$. From this bijection, we can extract a partition
$\set{N(c) \midd c \in S}$ of~$N$ by defining $N(c)$ as the set of voters that are
mapped to an element of the set $\{c^1, c^2, \ldots c^{\frac{n}{k}}\}$, for each
$c \in S$. It is easily verified that this partition satisfies the
conditions in \defref{def:pr1}. Therefore, $S$ provides perfect
representation for $(A, k)$. 
\end{proof}

Since EJR is incompatible with PR (see \exref{ex:4-2}), \maxP fails~EJR. However, it satisfies PJR.

\begin{theorem} \label{thm:lexP-PJR}
	\maxP satisfies PJR.
\end{theorem}

\begin{proof} 
We introduce one new piece of notation for this proof. For a committee $S\subseteq C$, let~$x^S$ be a leximax-optimal load distribution, given that $S$ is selected. As usual, we let $\bar x_i^S = \sum_{c \in S} x_{i,c}^S$.
	
Consider an instance $(A,k)$ and a committee $S$ output by \maxP. Assume that $S$ does not satisfy PJR. That is, there exists $\ell>0$ and a group $N^* \subseteq N$ of voters with $|N^*| \ge \ell n/k$, $|\bigcap_{i \in N^*}A_i| \ge \ell$ and $|S \cap (\bigcup_{i \in N^*} A_i)| \le \ell-1$. Note that there must exist a candidate $c^* \in \cap_{i \in N^*}A_i \setminus S$.
		
The average load among the voters in $N^*$ is
		\begin{equation}\label{maxP-pjr}
			\frac{1}{|N^*|}\sum_{i \in N^*} \bar x_i^S 
			\le \frac{|S \cap (\bigcup_{i \in N^*} A_i)|}{|N^*|} 
			\le \frac{\ell-1}{|N^*|}
			\le \frac{k}{n}-\frac{1}{|N^*|}.
		\end{equation}
Further, since the average load among voters in $N^*$ is strictly less than $\frac{k}{n}$ and the total load among all $n$ voters is $k$, the average load among voters in $N \setminus N^*$ is strictly greater than $\frac{k}{n}$. In particular, consider a leximax-optimal load distribution $x^S$ and let~$i'$ be a voter with maximum load among all voters in $N \backslash N^*$ according to $x^S$. It must be the case that this voter has load $\bar x_{i'}^S > \frac{k}{n}$. 
	
We can now complete the proof by constructing a committee which has a leximax-smaller vector of voter loads than $S$, contradicting the optimality of $S$. Consider a candidate $c$ with $x_{i',c}^S>0$. Such a candidate must exist because $\bar x_{i'}^S >0$. 
Consider replacing $c$ by $c^*$ to form committee $S'=S \cup \{ c^* \} \setminus \{ c \}$. We construct a valid load distribution $y$ for committee $S'$ as follows. Distribute the load of $c^*$ among voters in $N^*$ only in such a way that for each $i \in N^*$, $y_{i,c} \le \max (\frac{k}{n}-\bar x_i^S,0)$. This is possible because $\sum_{j \in N^*} \max (\frac{k}{n}-\bar x_j^S,0) \ge \sum_{j \in N^*} (\frac{k}{n}-\bar x_j^S) \ge 1$, where the last inequality follows from \eqref{maxP-pjr}. Setting $y_{i,c'}=x^S_{i,c'}$ for every voter $i$ and every candidate $c' \in S' \cap S$ yields
\begin{align*}
	&\bar y_i \le \bar x_i^S + \max \left(\frac{k}{n}-\bar x_i^S,0 \right) = \max \left(\frac{k}{n}, \bar x_i^S \right)  \text{ for all } i \in N^*\\
	&\bar y_{i'} < \bar x^S_{i'} \,\text{, and} \\
	&\bar y_i \le \bar x_i^S \, \text{ for all } i \in N \setminus N^*.
\end{align*}
In particular, since $\bar x^S_{i'}>\frac{k}{n}$ and $\bar y_i \le \frac{k}{n}$ for all $i$ with $\bar y_i > \bar x^S_i$, $y$ is a leximax-smaller vector of loads than $x^S$, contradicting optimality of $S$.
\end{proof}

\begin{remark}\label{rem:priceable2}
As is the case for \seqP, \maxP also satisfies priceability~\citep{peters2019proportionality} and therefore IPSC (see Remark \ref{rem:priceable}).
\end{remark}

\begin{corollary} \label{cor:maxP-JR}
	\maxP satisfies JR.
\end{corollary}

We note that \exref{ex:ex1} shows that simply minimizing the maximum voter load (without leximax tie-breaking) does not even yield committees satisfying JR.

\subsection{Results for \varP}

The proof of \thmref{thm:maxP-PR} directly applies to \varP.

\begin{theorem}\label{thm:varP-PR}
	\varP satisfies PR. 
\end{theorem}

Unlike \maxP, \varP fails PJR.

\begin{example} \label{ex:varP-PJR}
Let $C = \set{a,b,c,d,e,f,g}$, $k=6$, and consider the following profile with 100 voters:  
67 voters approve $\set{a,b,c,d}$,
12 voters approve $\set{e}$,
11 voters approve $\set{f}$, and
10 voters approve $\set{g}$.
Let $N^*$ be the set of voters approving $\set{a,b,c,d}$. We have $|N^*| = 67 \ge 4 \frac{n}{k}$ and $|\bigcap_{i \in N^*} A_i| = 4$. Thus, PJR requires that all four candidates in $\bigcap_{i \in N^*} A_i = \set{a,b,c,d}$  are selected. However, \varP selects $\set{a,b,c,e,f,g}$.
\end{example}

The previous example also shows that the sequential version of \varP violates~PJR.
Finally, we show that \varP satisfies JR.

	\begin{theorem}
		\varP satisfies JR.\label{thm:var-jr}
	\end{theorem}

The proof of \thmref{thm:var-jr} can be found in the appendix.

\section{Relationship to Apportionment Methods}
\label{sec:apportionment}

As mentioned in \secref{sec:related}, the well-studied apportionment problem~\citep{BaYo82a} constitutes a special case of approval-based committee elections. To see this, define a \textit{party-list profile} as a preference profile $A=(A_1, \dots, A_n)$ for which the set $C$ of candidates can be partitioned into ``parties'' $C = P_1 \mathbin{\dot{\cup}} P_2 \mathbin{\dot{\cup}} \ldots \mathbin{\dot{\cup}} P_p$ in such a way that 
each party $P_j$ contains at least $k$ candidates and 
each voter approves precisely the candidates of one party (i.e., for all $i \in N$, there exists a~$j\in\{1,\dots,p\}$ such that $A_i=P_j$). Each party-list profile $A$ can be summarized by a \textit{vote vector} $V_A=(v_1,  \ldots, v_p)$, where $v_j= |\{i \in N \midd A_i=P_j\}|$ is the total number of votes for party~$P_j$. 
An \textit{apportionment method} is a function that maps a vote vector $V=(v_1, \ldots, v_p)$ and a natural number $k$ to a \emph{seat distribution} \mbox{$z=(z_1, \ldots, z_p) \in \mathbb{N}_0^p$} with \mbox{$\sum_{j=1}^p z_j = k$}. 
Since vote vectors correspond to party-list profiles, approval-based committee voting rules are generalizations of apportionment methods. 
As a consequence, every approval-based committee voting rule~$\mathcal{R}$ induces an apportionment method~$M_{\mathcal{R}}$~\citep{BLS18a}: 
The number $z_j$ of seats that $M_{\mathcal{R}}$ allocates to a party $P_j$ is given by the number $|S \cap P_j|$ of candidates from party $P_j$ that are members of the committee $S$ selected by the rule~$\mathcal{R}$.

Apportionment methods have been extensively studied by  \mbox{\citet{BaYo82a}} and \citet{Puke14a}. Three of the most widely-used apportionment methods are 
\begin{itemize}
\item the \textit{D'Hondt method} (aka \textit{Jefferson method} or \textit{greatest divisors method}),
\item the \textit{Sainte-Laguë method} (aka \textit{Webster method} or \textit{major fractions method}), 
and 
\item the \textit{largest remainder method} (aka \textit{Hamilton method} or \textit{Hare--Niemeyer method}). 
\end{itemize}
Interestingly, all three apportionment methods are induced by different variants of \phrag's methods: 
\seqP and \maxP both induce the D'Hondt method~\citep{Phra95a,Jans16a,BLS18a}, 
\varP induces the Sainte-Laguë method~\citep{BLS18a},
and \eneP (using the Hare quota $q_H$) induces the largest remainder method~\citep{camps2019method}.%
\footnote{Under the assumption that there are at least as many seats as there are parties (i.e., $k\ge p$), the optimization variant that maximizes the minimum voter load (see Remark \ref{rem:minimaxP}) induces 
the \textit{Adams method}. This was remarked by \citet{Jans16a} and also follows from Proposition 3.11 of \citet{BaYo82a}.}

Some of the representation axioms discussed in \secref{sec:representation} have analogies in the apportionment literature: When restricted to party-list profiles, both EJR and PJR (see \defref{def:jr}) coincide with the requirement that the seat distribution satisfies \textit{lower quota} (i.e., $z_j \ge \lfloor k \frac{v_j}{n} \rfloor$ for all $j$).
Therefore, an apportionment method $M_\mathcal{R}$ induced by an approval-based committee voting $\mathcal{R}$ satisfies lower quota whenever $\mathcal{R}$ satisfies PJR. This observation, which was first made by \citet{BLS18a}, gives rise to an alternative proof for the fact that \varP fails PJR: \varP induces the Sainte-Laguë method~\citep{BLS18a}, which is well-known to fail lower quota \citep[p.~130]{BaYo82a}.\footnote{Indeed, the profile in \exref{ex:varP-PJR} is a party-list profile with vote vector $(67,12,11,10)$, for which the Sainte-Laguë method fails lower quota for $k=6$.}

Two further properties that are often studied in the apportionment setting are house monotonicity and population monotonicity \citep[p.~117]{BaYo82a}. 
\textit{House monotonicity} prescribes that no party loses seats when the house size is increased; this directly corresponds to \textit{committee monotonicity} for approval-based committee voting rules.  
Whereas \seqP satisfies committee monotonicity by definition, the non-sequential variants \maxP and \varP fail the property. This is implicit already in \phrag's 1896 paper~\citep{Phra96a}, and stated explicitly in the paper by \citet{MoOl15a};
here is a simple example.

\begin{example}\label{ex:com-mon}
Let $C = \{a, b, c\}$ and 
consider the following profile with $10$ voters: 
\[ 2 \times \{ a\} \qquad 3 \times \{a,c\} \qquad 3 \times \{b,c\} \qquad 2 \times \{b\}\]

Both \maxP and \varP select $\{c\}$ for $k=1$ and $\{a,b\}$ for $k=2$.
\end{example}

The D'Hondt method and
the Sainte-Laguë method satisfy house monotonicity \citep[p.~100]{BaYo82a}.
Consequently, \maxP and \varP satisfy committee monotonicity on party-list profiles.
In contrast, the largest remainder method fails house monotonicity and, therefore, \eneP fails committee monotonicity even on party-list profiles.

\textit{Population monotonicity} prescribes that, if the ratio $\frac{v_i}{v_j}$ increases, then it should not be the case that~$z_i$ decreases and $z_j$ increases. % \citep[p.~117]{BaYo82a}.
Population monotonicity is satisfied by the D'Hondt method and
the Sainte-Laguë method, but not by the largest remainder method~\citep[p.~117]{BaYo82a}.
We are not aware of a direct generalization of this property to approval-based committee voting rules; however, it is similar in spirit to \textit{support monotonicity}, introduced by \citet{SaFi17a}, who showed positive results for \seqP and \maxP.

\section{Conclusion}

We have shown that \phrag's load-balancing methods satisfy interesting representation axioms. In particular, the polynomial-time computable variant \seqP satisfies PJR.
Moreover, both \maxP and \varP satisfy PR and \maxP additionally satisfies PJR. 
Arguably, \maxP is the first known example of a ``natural'' rule satisfying both PR and PJR---the only other rule known to satisfy these two properties is an artificial construct that returns a PR committee if one exists and otherwise runs PAV~\citep{SFF+17a}. 

Since \seqP violates EJR, it remains an open problem whether EJR is compatible with committee monotonicity.\footnote{In the approval-based apportionment setting, where candidates can obtain multiple seats in the committee, EJR and committee monotonicity can be achieved simultaneously~\citep{BGP+22a}.} 
Further, the intricate nature of \exref{ex:seqP-EJR} seems to suggest that instances on which \seqP violates EJR are rare. It would be interesting to see whether \seqP satisfies EJR for realistic distributions of preferences and/or for reasonable domain restrictions.\footnote{Recent experimental work by \citet{BredereckF0N19} showed that committees satisfying JR very often satisfy EJR as well, supporting the hypothesis that instances for which \seqP fails EJR are  rare. An overview of domain restrictions for approval preferences can be found in the survey by \citet{ElkindEtAlTRENDS2017}.}
Finally, it would be of great interest to find axiomatic characterizations of \phrag's rules, i.e., to find sets of axiomatic properties that uniquely define \maxP, \varP, \seqP, and \eneP.

\begin{acks}
We would like to thank Xavier Mora for many fruitful discussions and for providing us with copies of the original papers by \phrag.
We also thank Marie-Louise Lackner for pointing out essential literature and providing us with translations. Furthermore, we thank Vincent Conitzer, Edith Elkind, Dominik Peters, Jannik Peters, Luis S{\'a}nchez-Fern{\'a}ndez, and Piotr Skowron for helpful comments. 
We are thankful to the \emph{Institut Mittag-Leffler} for permitting the use of \phrag{}'s photograph.

This material is based on work supported by ERC-StG 639945, NSF IIS-1527434 and ARO W911NF-12-1-0550, by a Feodor Lynen return fellowship of the Alexander von Humboldt Foundation, by COST Action IC1205 on Computational Social Choice, by a grant from the Knut and Alice Wallenberg Foundation, by the Isaac Newton Institute for Mathematical Sciences (EPSRC Grant Number EP/K032208/1), by a grant from the Simons foundation, by the Deutsche Forschungsgemeinschaft (DFG) under grant BR 4744/2-1, and by the Austrian Science Foundation FWF, grant P31890.
\end{acks}

\bibliographystyle{abbrvnat}

\begin{thebibliography}{65}
\providecommand{\natexlab}[1]{#1}
\providecommand{\url}[1]{\texttt{#1}}
\expandafter\ifx\csname urlstyle\endcsname\relax
  \providecommand{\doi}[1]{doi: #1}\else
  \providecommand{\doi}{doi: \begingroup \urlstyle{rm}\Url}\fi

\bibitem[Aziz and Lee(2020)]{AzLe20a}
H.~Aziz and B.~E. Lee.
\newblock The expanding approvals rule: improving proportional representation
  and monotonicity.
\newblock \emph{Social Choice and Welfare}, 54:\penalty0 1--45, 2020.

\bibitem[Aziz and Lee(2021)]{aziz2021proportionally}
H.~Aziz and B.~E. Lee.
\newblock Proportionally representative participatory budgeting with ordinal
  preferences.
\newblock In \emph{Proceedings of the 35th AAAI Conference on Artificial
  Intelligence (AAAI)}, pages 5110--5118. AAAI Press, 2021.

\bibitem[Aziz and Shah(2021)]{AZSh20a}
H.~Aziz and N.~Shah.
\newblock Participatory budgeting: Models and approaches.
\newblock In T.~Rudas and G.~P{\'e}li, editors, \emph{Pathways Between Social
  Science and Computational Social Science}, pages 215--236. Springer, 2021.

\bibitem[Aziz et~al.(2015)Aziz, Gaspers, Gudmundsson, Mackenzie, Mattei, and
  Walsh]{AGG+15a}
H.~Aziz, S.~Gaspers, J.~Gudmundsson, S.~Mackenzie, N.~Mattei, and T.~Walsh.
\newblock Computational aspects of multi-winner approval voting.
\newblock In \emph{Proceedings of the 14th International Conference on
  Autonomous Agents and Multiagent Systems (AAMAS)}, pages 107--115. IFAAMAS,
  2015.

\bibitem[Aziz et~al.(2017)Aziz, Brill, Conitzer, Elkind, Freeman, and
  Walsh]{justifiedRepresentation}
H.~Aziz, M.~Brill, V.~Conitzer, E.~Elkind, R.~Freeman, and T.~Walsh.
\newblock Justified representation in approval-based committee voting.
\newblock \emph{Social Choice and Welfare}, 48\penalty0 (2):\penalty0 461--485,
  2017.

\bibitem[Aziz et~al.(2018{\natexlab{a}})Aziz, Elkind, Huang, Lackner,
  S{\'a}nchez-Fern{\'a}ndez, and Skowron]{AEHLSS-polyejr}
H.~Aziz, E.~Elkind, S.~Huang, M.~Lackner, L.~S{\'a}nchez-Fern{\'a}ndez, and
  P.~Skowron.
\newblock On the complexity of extended and proportional justified
  representation.
\newblock In \emph{Proceedings of the 32nd AAAI Conference on Artificial
  Intelligence (AAAI)}, pages 902--909. AAAI Press, 2018{\natexlab{a}}.

\bibitem[Aziz et~al.(2018{\natexlab{b}})Aziz, Lee, and
  Talmon]{aziz2018proportionally}
H.~Aziz, B.~E. Lee, and N.~Talmon.
\newblock Proportionally representative participatory budgeting: Axioms and
  algorithms.
\newblock In \emph{Proceedings of the 17th International Conference on
  Autonomous Agents and Multiagent Systems (AAMAS)}, pages 23--31. IFAAMAS,
  2018{\natexlab{b}}.

\bibitem[Balinski and Young(1982)]{BaYo82a}
M.~Balinski and H.~P. Young.
\newblock \emph{Fair Representation: {M}eeting the Ideal of One Man, One Vote}.
\newblock Yale University Press, 1982.
\newblock (2nd Edition [with identical pagination], Brookings Institution
  Press, 2001).

\bibitem[Betzler et~al.(2013)Betzler, Slinko, and Uhlmann]{BSU13a}
N.~Betzler, A.~Slinko, and J.~Uhlmann.
\newblock On the computation of fully proportional representation.
\newblock \emph{Journal of Artificial Intelligence Research}, 47:\penalty0
  475--519, 2013.

\bibitem[Bredereck et~al.(2019)Bredereck, Faliszewski, Kaczmarczyk, and
  Niedermeier]{BredereckF0N19}
R.~Bredereck, P.~Faliszewski, A.~Kaczmarczyk, and R.~Niedermeier.
\newblock An experimental view on committees providing justified
  representation.
\newblock In \emph{Proceedings of the 28th International Joint Conference on
  Artificial Intelligence (IJCAI)}, pages 109--115. IJCAI, 2019.

\bibitem[Brent and Zimmermann(2010)]{brent2010modern}
R.~P. Brent and P.~Zimmermann.
\newblock \emph{Modern computer arithmetic}, volume~18.
\newblock Cambridge University Press, 2010.

\bibitem[Brill et~al.(2017)Brill, Freeman, Janson, and
  Lackner]{BrillFJL17-phragmen}
M.~Brill, R.~Freeman, S.~Janson, and M.~Lackner.
\newblock Phragm{\'e}n's voting methods and justified representation.
\newblock In \emph{Proceedings of the 31st AAAI Conference on Artificial
  Intelligence (AAAI)}, pages 406--413. AAAI Press, 2017.

\bibitem[Brill et~al.(2018)Brill, Laslier, and Skowron]{BLS18a}
M.~Brill, J.-F. Laslier, and P.~Skowron.
\newblock Multiwinner approval rules as apportionment methods.
\newblock \emph{Journal of Theoretical Politics}, 30\penalty0 (3):\penalty0
  358--382, 2018.

\bibitem[Brill et~al.(2022)Brill, G{\"o}lz, Peters, Schmidt-Kraepelin, and
  Wilker]{BGP+22a}
M.~Brill, P.~G{\"o}lz, D.~Peters, U.~Schmidt-Kraepelin, and K.~Wilker.
\newblock Approval-based apportionment.
\newblock \emph{Mathematical Programming}, 2022.
\newblock \doi{10.1007/s10107-022-01852-1}.
\newblock Forthcoming.

\bibitem[Burdges et~al.(2020)Burdges, Cevallos, Czaban, Habermeier, Hosseini,
  Lama, Alper, Luo, Shirazi, Stewart, and Wood]{polkadot-overview}
J.~Burdges, A.~Cevallos, P.~Czaban, R.~Habermeier, S.~Hosseini, F.~Lama, H.~K.
  Alper, X.~Luo, F.~Shirazi, A.~Stewart, and G.~Wood.
\newblock Overview of {P}olkadot and its design considerations.
\newblock Technical report, arXiv:2005.13456 [cs.CR], 2020.

\bibitem[Cairns(1924)]{imc-toronto}
W.~D. Cairns.
\newblock The {I}nternational {M}athematical {C}ongress at {T}oronto.
\newblock \emph{The American Mathematical Monthly}, 31\penalty0 (9):\penalty0
  411--417, 1924.

\bibitem[Camps et~al.(2019)Camps, Mora, and Saumell]{camps2019method}
R.~Camps, X.~Mora, and L.~Saumell.
\newblock The method of {E}nestr\"om and {P}hragm{\'e}n for parliamentary
  elections by means of approval voting.
\newblock Technical report, arXiv:1907.10590 [econ.TH], 2019.

\bibitem[Cevallos and Stewart(2021)]{CeSt21a}
A.~Cevallos and A.~Stewart.
\newblock A verifiably secure and proportional committee election rule.
\newblock In \emph{Proceedings of the 3rd ACM Conference on Advances in
  Financial Technologies (AFT)}, pages 29--42. ACM, 2021.

\bibitem[Elkind et~al.(2017{\natexlab{a}})Elkind, Faliszewski, Skowron, and
  Slinko]{elk-fal-sko-sli:c:multiwinner-rules}
E.~Elkind, P.~Faliszewski, P.~Skowron, and A.~Slinko.
\newblock Properties of multiwinner voting rules.
\newblock \emph{Social Choice and Welfare}, 48\penalty0 (3):\penalty0 599--632,
  2017{\natexlab{a}}.

\bibitem[Elkind et~al.(2017{\natexlab{b}})Elkind, Lackner, and
  Peters]{ElkindEtAlTRENDS2017}
E.~Elkind, M.~Lackner, and D.~Peters.
\newblock Structured preferences.
\newblock In U.~Endriss, editor, \emph{Trends in Computational Social Choice},
  chapter~10, pages 187--207. AI Access, 2017{\natexlab{b}}.

\bibitem[Faliszewski et~al.(2017)Faliszewski, Skowron, Slinko, and
  Talmon]{FSST17a}
P.~Faliszewski, P.~Skowron, A.~Slinko, and N.~Talmon.
\newblock Multiwinner voting: A new challenge for social choice theory.
\newblock In U.~Endriss, editor, \emph{Trends in Computational Social Choice},
  chapter~2. AI Access, 2017.

\bibitem[Garey and Johnson(1979)]{GaJo79a}
M.~R. Garey and D.~S. Johnson.
\newblock \emph{Computers and Intractability: A Guide to the Theory of
  NP-Completeness}.
\newblock W. H. Freeman, 1979.

\bibitem[Garey et~al.(1976)Garey, Johnson, and Stockmeyer]{GJS76a}
M.~R. Garey, D.~S. Johnson, and L.~J. Stockmeyer.
\newblock Some simplified {NP}-complete graph problems.
\newblock \emph{Theor. Comput. Sci.}, 1\penalty0 (3):\penalty0 237--267, 1976.

\bibitem[Israel and Brill(2021)]{IsBr21b}
J.~Israel and M.~Brill.
\newblock Dynamic proportional rankings.
\newblock In \emph{Proceedings of the 30th International Joint Conference on
  Artificial Intelligence (IJCAI)}, pages 261--267. IJCAI, 2021.

\bibitem[Janson(2012)]{Jans12a}
S.~Janson.
\newblock Proportionella valmetoder.
\newblock Unpublished manuscript. Available at
  \url{http://www2.math.uu.se/~svante/papers/sjV6.pdf}, 2012.

\bibitem[Janson(2018{\natexlab{a}})]{Jans16a}
S.~Janson.
\newblock Phragm{\'e}n's and {T}hiele's election methods.
\newblock Technical report, arXiv:1611.08826v2 [math.HO], 2018{\natexlab{a}}.

\bibitem[Janson(2018{\natexlab{b}})]{janson2018thresholds}
S.~Janson.
\newblock Thresholds quantifying proportionality criteria for election methods.
\newblock Technical report, arXiv:1810.06377 [cs.GT], 2018{\natexlab{b}}.

\bibitem[Jaworski and Skowron(2022)]{JaSk22a}
M.~Jaworski and P.~Skowron.
\newblock Phragm{\'e}n rules for degressive and regressive proportionality.
\newblock In \emph{Proceedings of the 31st International Joint Conference on
  Artificial Intelligence (IJCAI)}, pages 328--334. IJCAI, 2022.

\bibitem[Kilgour(2010)]{Kilg10a}
D.~M. Kilgour.
\newblock Approval balloting for multi-winner elections.
\newblock In \emph{Handbook on Approval Voting}, chapter~6. Springer, 2010.

\bibitem[Lackner and Maly(2023)]{LaMa23a}
M.~Lackner and J.~Maly.
\newblock Proportional decisions in perpetual voting.
\newblock In \emph{Proceedings of the 37th AAAI Conference on Artificial
  Intelligence (AAAI)}. AAAI Press, 2023.
\newblock Forthcoming.

\bibitem[Lackner and Skowron(2022)]{LaSk21b}
M.~Lackner and P.~Skowron.
\newblock \emph{Multi-Winner Voting with Approval Preferences}.
\newblock Springer, 2022.

\bibitem[Lu and Boutilier(2011)]{LuBo11d}
T.~Lu and C.~Boutilier.
\newblock Budgeted social choice: From consensus to personalized decision
  making.
\newblock In \emph{Proceedings of the 22nd International Joint Conference on
  Artificial Intelligence (IJCAI)}, pages 280--286. AAAI Press, 2011.

\bibitem[M{\"o}ller(2008)]{moller2008schonhage}
N.~M{\"o}ller.
\newblock On {S}ch{\"o}nhage's algorithm and subquadratic integer {GCD}
  computation.
\newblock \emph{Mathematics of Computation}, 77\penalty0 (261):\penalty0
  589--607, 2008.

\bibitem[Monroe(1995)]{Monr95a}
B.~L. Monroe.
\newblock Fully proportional representation.
\newblock \emph{The American Political Science Review}, 89\penalty0
  (4):\penalty0 925--940, 1995.

\bibitem[Mora(2016)]{Mora16a}
X.~Mora.
\newblock Phragm{\'e}n's sequential method with a variance criterion.
\newblock Technical report, arXiv:1611.06833 [math.OC], 2016.

\bibitem[Mora and Oliver(2015)]{MoOl15a}
X.~Mora and M.~Oliver.
\newblock Eleccions mitjan{\c c}ant el vot d'aprovaci{\'o}. {E}l m{\`e}tode de
  {P}hragm{\'e}n i algunes variants.
\newblock \emph{Butllet{\'\i} de la Societat Catalana de Matem{\`a}tiques},
  30\penalty0 (1):\penalty0 57--101, 2015.

\bibitem[Moulin(1988)]{Moul88a}
H.~Moulin.
\newblock \emph{Axioms of Cooperative Decision Making}.
\newblock Cambridge University Press, 1988.

\bibitem[Ogryczak(1997)]{Ogry97a}
W.~Ogryczak.
\newblock On the lexicographic minimax approach to location problems.
\newblock \emph{European Journal of Operational Research}, 100\penalty0
  (3):\penalty0 566--585, 1997.

\bibitem[Ogryczak and Sliwinski(2006)]{DBLP:conf/iccsa/OgryczakS06}
W.~Ogryczak and T.~Sliwinski.
\newblock On direct methods for lexicographic min-max optimization.
\newblock In M.~L. Gavrilova, O.~Gervasi, V.~Kumar, C.~J.~K. Tan, D.~Taniar,
  A.~Lagan{\`{a}}, Y.~Mun, and H.~Choo, editors, \emph{Computational Science
  and Its Applications - {ICCSA} 2006}, volume 3982 of \emph{Lecture Notes in
  Computer Science}, pages 802--811. Springer, 2006.

\bibitem[Peters and Skowron(2020)]{peters2019proportionality}
D.~Peters and P.~Skowron.
\newblock Proportionality and the limits of welfarism.
\newblock In \emph{Proceedings of the 21st ACM Conference on Economics and
  Computation (ACM-EC)}, pages 793--794, 2020.
\newblock Full version arXiv:1911.11747 [cs.GT].

\bibitem[Peters et~al.(2021)Peters, Pierczy{\'n}ski, and Skowron]{PPP21a}
D.~Peters, G.~Pierczy{\'n}ski, and P.~Skowron.
\newblock Proportional participatory budgeting with additive utilities.
\newblock In \emph{Proceedings of the 35th Annual Conference on Neural
  Information Processing Systems (NeurIPS)}, pages 12726--12737, 2021.

\bibitem[Phragm{\'e}n(1893)]{Phrag93}
E.~Phragm{\'e}n.
\newblock Om proportionella val.
\newblock \emph{Stockholms Dagblad}, 14 March 1893, 1893.
\newblock Summary of a public lecture published in a newspaper.

\bibitem[Phragm{\'e}n(1894)]{Phra94a}
E.~Phragm{\'e}n.
\newblock Sur une m{\'e}thode nouvelle pour r{\'e}aliser, dans les
  {\'e}lections, la repr{\'e}sentation proportionnelle des partis.
\newblock \emph{\"Ofversigt af Kongliga Vetenskaps-Akademiens F\"orhandlingar},
  51\penalty0 (3):\penalty0 133--137, 1894.

\bibitem[Phragm{\'e}n(1895)]{Phra95a}
E.~Phragm{\'e}n.
\newblock \emph{Proportionella val. En valteknisk studie}.
\newblock Svenska sp{\"o}rsm{\aa}l 25. Lars H{\"o}kersbergs f{\"o}rlag,
  Stockholm, 1895.

\bibitem[Phragm{\'e}n(1896)]{Phra96a}
E.~Phragm{\'e}n.
\newblock Sur la th{\'e}orie des {\'e}lections multiples.
\newblock \emph{{\"O}fversigt af Kongliga Vetenskaps-Akademiens
  F{\"o}rhandlingar}, 53:\penalty0 181--191, 1896.

\bibitem[Phragm{\'e}n(1899)]{Phra99a}
E.~Phragm{\'e}n.
\newblock Till fr{\aa}gan om en proportionell valmetod.
\newblock \emph{Statsvetenskaplig Tidskrift}, 2\penalty0 (2):\penalty0
  297--305, 1899.

\bibitem[Phragm{\'e}n and Lindel{\"o}f(1908)]{phragmen1908extension}
E.~Phragm{\'e}n and E.~Lindel{\"o}f.
\newblock Sur une extension d'un principe classique de l'analyse et sur
  quelques propri{\'e}t{\'e}s des fonctions monog{\`e}nes dans le voisinage
  d'un point singulier.
\newblock \emph{Acta Mathematica}, 31\penalty0 (1):\penalty0 381--406, 1908.

\bibitem[Pia et~al.(2017)Pia, Dey, and Molinaro]{PDM17a}
A.~D. Pia, S.~S. Dey, and M.~Molinaro.
\newblock Mixed-integer quadratic programming is in {NP}.
\newblock \emph{Mathematical Programming}, 162:\penalty0 225--240, 2017.

\bibitem[{Polkadot Wiki}(2021)]{Polk21a}
{Polkadot Wiki}.
\newblock {NPoS} election algorithms.
\newblock \url{https://wiki.polkadot.network/docs/learn-phragmen}, 2021.
\newblock Accessed: {2023-01-01}.

\bibitem[Potthof and Brams(1998)]{PoBr98a}
R.~F. Potthof and S.~J. Brams.
\newblock Proportional representation: Broadening the options.
\newblock \emph{Journal of Theoretical Politics}, 10\penalty0 (2):\penalty0
  147--178, 1998.

\bibitem[Procaccia et~al.(2008)Procaccia, Rosenschein, and Zohar]{PSZ08a}
A.~D. Procaccia, J.~S. Rosenschein, and A.~Zohar.
\newblock On the complexity of achieving proportional representation.
\newblock \emph{Social Choice and Welfare}, 30:\penalty0 353--362, 2008.

\bibitem[Pukelsheim(2014)]{Puke14a}
F.~Pukelsheim.
\newblock \emph{Proportional Representation: Apportionment Methods and Their
  Applications}.
\newblock Springer, 2014.

\bibitem[Rosenfeld et~al.(2022)Rosenfeld, Shapiro, and Talmon]{RST22a}
A.~Rosenfeld, E.~Shapiro, and N.~Talmon.
\newblock Proportional ranking in primary elections: A case study.
\newblock \emph{Party Politics}, 2022.
\newblock \doi{10.1177/13540688211066711}.
\newblock Forthcoming.

\bibitem[S{\'a}nchez-Fern{\'a}ndez and Fisteus(2019)]{SaFi17a}
L.~S{\'a}nchez-Fern{\'a}ndez and J.~A. Fisteus.
\newblock Monotonicity axioms in approval-based multi-winner voting rules.
\newblock In \emph{Proceedings of the 18th International Conference on
  Autonomous Agents and Multiagent Systems (AAMAS)}, pages 485--493. IFAAMAS,
  2019.
\newblock Full version arXiv:1710.04246v3 [cs.GT].

\bibitem[S{\'a}nchez-Fern{\'a}ndez
  et~al.(2017{\natexlab{a}})S{\'a}nchez-Fern{\'a}ndez, Elkind, and
  Lackner]{SEL17a}
L.~S{\'a}nchez-Fern{\'a}ndez, E.~Elkind, and M.~Lackner.
\newblock Committees providing {EJR} can be computed efficiently.
\newblock Technical report, arXiv:1704.00356v3 [cs.GT], 2017{\natexlab{a}}.

\bibitem[S{\'a}nchez-Fern{\'a}ndez
  et~al.(2017{\natexlab{b}})S{\'a}nchez-Fern{\'a}ndez, Elkind, Lackner,
  Fern{\'a}ndez, Fisteus, {Basanta Val}, and Skowron]{SFF+17a}
L.~S{\'a}nchez-Fern{\'a}ndez, E.~Elkind, M.~Lackner, N.~Fern{\'a}ndez, J.~A.
  Fisteus, P.~{Basanta Val}, and P.~Skowron.
\newblock Proportional justified representation.
\newblock In \emph{Proceedings of the 31st AAAI Conference on Artificial
  Intelligence (AAAI)}, pages 670--676. AAAI Press, 2017{\natexlab{b}}.

\bibitem[S{\'a}nchez-Fern{\'a}ndez et~al.(2022)S{\'a}nchez-Fern{\'a}ndez,
  Fern{\'a}ndez, Fisteus, and Brill]{SFFB22b}
L.~S{\'a}nchez-Fern{\'a}ndez, N.~Fern{\'a}ndez, J.~A. Fisteus, and M.~Brill.
\newblock The maximin support method: An extension of the {D'Hondt} method to
  approval-based multiwinner elections.
\newblock \emph{Mathematical Programming}, 2022.
\newblock \doi{10.1007/s10107-022-01805-8}.
\newblock Forthcoming.

\bibitem[Schmeidler(1969)]{Schm69a}
D.~Schmeidler.
\newblock The nucleolus of a characteristic function game.
\newblock \emph{SIAM Journal on Applied Mathematics}, 17\penalty0 (6):\penalty0
  1163--1170, 1969.

\bibitem[Schrijver(1986)]{Schr86a}
A.~Schrijver.
\newblock \emph{Theory of Linear and Integer Programming}.
\newblock John Wiley \& Sons, 1986.

\bibitem[Skowron et~al.(2016)Skowron, Faliszewski, and Lang]{owaWinner}
P.~Skowron, P.~Faliszewski, and J.~Lang.
\newblock Finding a collective set of items: From proportional
  multirepresentation to group recommendation.
\newblock \emph{Artificial Intelligence}, 241:\penalty0 191--216, 2016.

\bibitem[Skowron et~al.(2017)Skowron, Lackner, Brill, Peters, and
  Elkind]{SLB+17a}
P.~Skowron, M.~Lackner, M.~Brill, D.~Peters, and E.~Elkind.
\newblock Proportional rankings.
\newblock In \emph{Proceedings of the 26th International Joint Conference on
  Artificial Intelligence (IJCAI)}, pages 409--415. IJCAI, 2017.

\bibitem[Strassen(1971)]{schonhage1971schnelle}
A.~S.~V. Strassen.
\newblock Schnelle {M}ultiplikation gro\ss er {Z}ahlen.
\newblock \emph{Computing}, 7\penalty0 (3-4):\penalty0 281--292, 1971.

\bibitem[Stubhaug(2010)]{stubhaug2010gosta}
A.~Stubhaug.
\newblock \emph{G{\"o}sta Mittag-Leffler: A man of conviction}.
\newblock Springer Science \& Business Media, 2010.

\bibitem[Thiele(1895)]{Thie95a}
T.~N. Thiele.
\newblock Om flerfoldsvalg.
\newblock \emph{Oversigt over det Kongelige Danske Videnskabernes Selskabs
  Forhandlinger}, pages 415--441, 1895.

\bibitem[Tideman(1995)]{Tide95a}
N.~Tideman.
\newblock The single transferable vote.
\newblock \emph{Journal of Economic Perspectives}, 9\penalty0 (1):\penalty0
  27--38, 1995.

\end{thebibliography}

\appendix
\section{Appendix}

\subsection{Proof of \thmref{thm:var-jr}}

We first prove a lemma.

	\begin{lemma}
	Let $0<\alpha<1$ and $(x_i)_{1\le i \le n}$ be a sequence with $0\leq x_i\leq \alpha$ for all $i \in \{1, \ldots, n\}$ and $\sum_{i=1}^n x_i = 1$.
	Then, $\sum_{i=1}^n x_i^2 \leq \alpha$.
	\label{lem:var-bound}
	\end{lemma}	

	\begin{proof}
$\sum_{i=1}^n x_i^2 \leq \sum_{i=1}^n \alpha x_i = \alpha$. 
	\end{proof}

We can now prove \thmref{thm:var-jr}.

\begin{proof}[\thmref{thm:var-jr}]
		Consider an instance $(A,k)$ and a committee $S$ output by \varP. Assume that $S$ does not satisfy JR. That is, there exists a group $N^*$ with $|N^*| \ge \frac{n}{k}$, such that $\bigcap_{i \in N^*} A_i \neq \emptyset$ and $|S \cap (\bigcup_{i \in N^*} A_i)| = \emptyset$. Clearly, $|N^*|<n$.
	
		Let $i'$ be a voter with maximum load (i.e., $\bar x_{i'} \ge \bar x_i$ for all $i \in N$), and let $c$ be a candidate with $x_{i',c}>0$. Such a $c$ must exist because the total load on $i'$ is non-zero. 
	
		First note that the average load on voters in $N \backslash N^*$ is \[\frac{k}{|N \backslash N^*|} \ge \frac{k}{n-\frac{n}{k}} = \frac{k^2}{kn-n}.\] 
		Therefore, since~$i'$ is a voter with maximum load, it must be the case that $\bar x_{i'} \ge \frac{k^2}{kn-n}$. Further, $\bar x_i = \bar x_{i'}$ for all voters $i \in N_c$. If this were not the case for some voter $i$, it would be possible to decrease the variance of the load distribution by reducing $x_{i',c}$ by some small amount and increasing $x_{i,c}$ accordingly, thus reducing the difference between the loads on $i'$ and $i$ while leaving all other loads unchanged, which reduces the variance.
	
		Let $d \in \bigcap_{i \in N^*} A_i$, and let $T = S \cup \{d\} \setminus \{c\}$. That is, $T$ is the committee obtained by starting with $S$ and replacing $c$ with a candidate approved by all voters in $N^*$. To complete the proof, we consider the effect that this replacement has on the quanity $\sum_{i \in N} \bar x_i^2$, which is the objective minimized by \varP.
	
It is possible to distribute the load of candidate $d$ evenly across all
(previously unrepresented) voters in $N^*$. Therefore, the addition of $d$
contributes at most $\sum_{i \in N^*} \frac{1}{|N^*|^2} =
\frac{1}{|N^*|} \leq \frac{k}{n}$ to the objective. 
		On the other hand, removing $c$ from the committee decreases the objective by 
		\begin{align*}
			\sum_{i \in N_c} &(\bar x_i^2 - (\bar x_i - x_{i,c})^2) = \sum_{i \in N_c} (\bar x_i^2 - \bar x_i^2 + 2\bar x_i x_{i,c} - x_{i,c}^2)
			= \sum_{i \in N_c} (2 \bar x_ix_{i,c}-x_{i,c}^2)\\
			&= \sum_{i \in N_c} (2 \bar x_{i'}x_{i,c}-x_{i,c}^2)
			= 2 \bar x_{i'} - \sum_{i \in N_c} x_{i,c}^2
			\ge 2\bar x_{i'} - \bar x_{i'} 
			= \bar x_{i'}
			\ge \frac{k^2}{kn-n} > \frac{k}{n} \text,
		\end{align*}
		where the first inequality follows from Lemma~\ref{lem:var-bound}.
	Therefore, replacing $c$ by $d$ causes a net decrease to the objective, contradicting minimality of the variance of committee $S$. 
	We have thus obtained a contradiction to our assumption that $S$ does not provide JR. 
    \end{proof}

\subsection{\seqP violates EJR}
\label{app:EJRcalculation}

Table \ref{tbl:calc} shows the necessary calculations for computing \seqP in \exref{ex:seqP-EJR}.

\begin{table*}[htb] \footnotesize
\label{tbl:EJRcalculation}
\centering
\resizebox{\columnwidth}{!}{%
\begin{tabular}{ccccccccccccc} 
	\toprule
	$c$ & $s_c^{(1)}$ & $s_c^{(2)}$ & $s_c^{(3)}$ & $s_c^{(4)}$ & $s_c^{(5)}$ & $s_c^{(6)}$ & $s_c^{(7)}$ & $s_c^{(8)}$ & $s_c^{(9)}$ & $s_c^{(10)}$ & $s_c^{(11)}$ & $s_c^{(12)}$ \\ 
	\midrule
	$c_1$ & 0.125 & 0.163 & 0.2 & 0.238 & 0.275 & 0.310 & 0.332 & \textbf{0.369} & -- & -- & -- & -- \\
	$c_2$ & 0.077 & 0.119 & 0.162 & 0.204 & \textbf{0.246} & -- & -- & -- & -- & -- & -- & --\\
	$c_3$ & \textbf{0.05} & -- & -- & -- & -- & -- & -- & -- & -- & -- & -- & --\\
	$c_4$ & 0.05 & \textbf{0.1} & -- & -- & -- & -- & -- & -- & -- & -- & -- & --\\
	$c_5$ & 0.05 & 0.1 & \textbf{0.15}& -- & -- & -- & -- & -- & -- & -- & -- & --\\
	$c_6$ & 0.05 & 0.1 & 0.15 & \textbf{0.2} & -- & -- & -- & -- & -- & -- & -- & --\\
	$c_7$ & 0.05 & 0.1 & 0.15 & 0.2 & 0.25 & \textbf{0.275} & -- & -- & -- & -- & -- & --\\
	$c_8$ & 0.05 & 0.1 & 0.15 & 0.2 & 0.25 & 0.275 & \textbf{0.325} & -- & -- & -- & -- & --\\
	$c_9$ & 0.05 & 0.1 & 0.15 & 0.2 & 0.25 & 0.275 & 0.325 & 0.375 & \textbf{0.388} & -- & -- & -- \\ 
	$c_{10}$ & 0.05 & 0.1 & 0.15 & 0.2 & 0.25 & 0.275 & 0.325 & 0.375 & 0.388 & \textbf{0.438} & -- & --\\
	$c_{11}$ & 0.05 & 0.1 & 0.15 & 0.2 & 0.25 & 0.275 & 0.325 & 0.375 & 0.388 & 0.438 & \textbf{0.488} & --\\
	$c_{12}$ & 0.05 & 0.1 & 0.15 & 0.2 & 0.25 & 0.275 & 0.325 & 0.375 & 0.388 & 0.438 & 0.488 & \textbf{0.538}\\
	$a$ & 0.25 & 0.25 & 0.25 & 0.25 & 0.25 & 0.373 & 0.373 & 0.373 & 0.558 & 0.558 & 0.558 & 0.558\\
	$b$ & 0.25 & 0.25 & 0.25 & 0.25 & 0.25 & 0.373 & 0.373 & 0.373 & 0.558 & 0.558 & 0.558 & 0.558\\ \bottomrule
\end{tabular}
}
	\caption{The values $s_c^{(j)}$ (rounded to three decimal places) for each remaining candidate $c \in C$ and each round $j \in \{1, \ldots, 12\}$ in  Example~\ref{ex:seqP-EJR}. Entries in bold distinguish the candidate with lowest $s_c^{(j)}$ value (up to tie-breaking) in each round.}
	\label{tbl:calc}
	\end{table*}

\end{document}